%% file: Generalized_resistance_geometry.tex
\documentclass[onefignum,onetabnum]{siamart171218}

\usepackage{physics}
\usepackage{bm}
\usepackage{mathtools}
\usepackage{booktabs, multirow}
\usepackage{amsmath, amssymb, amsfonts, latexsym}
\usepackage{amscd}  
\usepackage{tikz-cd}
\usetikzlibrary{fit, backgrounds}
\usepackage{graphicx} 
\usepackage{enumitem}
\usepackage{comment}

\usepackage{tabularx}
\usepackage{array}
\usepackage{multirow}
\usepackage{booktabs}
\usepackage{tabulary}
\usepackage{booktabs}

\newcolumntype{Y}{>{\centering\arraybackslash}X}

\newsiamremark{example}{Example}

\usepackage{xcolor}
\usepackage{colortbl}
\usepackage{ulem} 

\setlist[enumerate,1]{label=\arabic*., ref=\arabic*., font=\normalfont}

\usepackage{tikz}
\usetikzlibrary{calc,arrows.meta}

\input{Generalized_resistance_geometry_shared.tex}

\ifpdf
\hypersetup{
  pdftitle={Summary},
  pdfauthor={Y. Kajiura}
}
\fi

\myexternaldocument{ex_supplement}

\renewcommand{\thesection}{\arabic{section}}
\renewcommand{\thesubsection}{\thesection.\arabic{subsection}}
\renewcommand{\thesubsubsection}{\thesubsection.\arabic{subsubsection}}

\makeatletter

\renewcommand{\tableofcontents}{%
  \par\bigskip
  \phantomsection
  \begin{center}\Large\bfseries Contents\end{center}%
  \par\medskip
  \@starttoc{toc}%
}

\newcommand{\mapsfrom}{\mathrel{\reflectbox{$\mapsto$}}}  

\begin{document}

\maketitle


\begin{abstract}
    We develop a generalized resistance geometry for unsigned directed graphs and signed undirected graphs based on Kron reduction and effective resistance. For unsigned strongly connected weight-balanced directed graphs, we establish a generalized Fiedler--Bapat identity involving an associated signed undirected Laplacian constructed from the symmetrized pseudoinverse. We identify this Laplacian as the sum of the symmetrized directed Laplacian and a positive-semidefinite correction induced by graph asymmetry. This decomposition yields an effective-resistance comparison that extends a previous normality-based result, and we prove that the construction commutes with Kron reduction. For general unsigned strongly connected directed graphs satisfying a positivity condition, we define resistance curvature and resistance radius through a weight-balanced representation. Motivated by the associated signed undirected graphs, we introduce a class whose effective resistances form a metric and show that the resulting resistance matrices are precisely the strict negative type metric matrices. Within this framework, we characterize the unique solution of the maximum graph-variance problem and develop generalized resistive embeddings into Euclidean space.
\end{abstract}



\section{Introduction}
\label{sec:introduction}

\subsection{Background}
\label{subsec:background}
    Kron reduction and effective resistance provide fundamental connections among spectral graph theory, electrical networks, and model reduction~\cite{Dorfler2013,Doyle_Snell_1984,Klein1993,Kron1939TAN}. For connected unsigned undirected graphs, effective resistance is a metric defined through the Laplacian pseudoinverse, while Kron reduction eliminates nodes without changing the effective resistances among the retained nodes~\cite{Dorfler2013,Klein1993}. The resulting resistance matrix also admits a geometric interpretation through Euclidean embeddings and simplex geometry, with applications to graph-variance optimization and related graph invariants~\cite{Devriendt2022graph,Fiedler2011_2}.
    
    These notions have recently been extended to directed graphs. For unsigned strongly connected directed graphs, reachability-based Kron reduction preserves directed effective resistance~\cite{Sugiyama2023graph}. In the weight-balanced case, the symmetrized pseudoinverse of the directed Laplacian is positive-semidefinite and represents effective resistance as a symmetric node metric \cite{Fontan2023pseudoinverse,Sugiyama2023graph}. Related studies have also investigated directed graph symmetrization and structural properties of Laplacian pseudoinverses \cite{Fitch2019EffectiveResistance,Fontan2021pseudoinverse,Fontan2023pseudoinverse}.
    
    In the classical unsigned undirected setting, the Fiedler--Bapat identity connects the resistance matrix with the graph Laplacian and gives rise to the resistance curvature and resistance radius \cite{Devriendt2022graph,Fiedler1962}.
    These quantities characterize the maximum graph-variance problem and admit natural interpretations in the associated resistive simplex.
    Extending this framework is nontrivial because directed Laplacians are asymmetric, while signed undirected Laplacians may induce effective resistances that violate the triangle inequality. In this paper, we address these difficulties through an associated signed undirected Laplacian for weight-balanced directed graphs and use it to extend resistance invariants, metric geometry, graph-variance optimization, and resistive embeddings beyond the classical setting.

\subsection{Contribution}
\label{subsec:contribution}
\begin{itemize}[leftmargin=1em,labelsep=0.5em]
    \item
        \textbf{Generalized Fiedler--Bapat identity for directed graphs.}
        We establish a generalized Fiedler--Bapat identity for strongly connected weight-balanced (SCWB) directed graphs, relating the resistance matrix to an associated signed undirected Laplacian constructed from the symmetrized pseudoinverse of the directed Laplacian.
        Furthermore, for general strongly connected (SC) directed graphs satisfying a positivity condition, we define resistance curvature and resistance radius through a weight-balanced representation.

    \item
        \textbf{Structural properties and Kron reduction compatibility.}
        We show that the associated signed undirected Laplacian decomposes into the symmetrized directed Laplacian and a positive-semidefinite correction. This decomposition yields a Loewner-order comparison between the corresponding pseudoinverses and, consequently, a comparison of their effective resistances.
        Crucially, we prove that the construction of the associated signed undirected Laplacian commutes with Kron reduction and summarize the resulting relations in a commutative diagram.

    \item
        \textbf{Generalized resistance metrics and geometric consequences.}
        We introduce a class of signed undirected Laplacians whose effective resistances form a metric and show that the induced resistance matrices are precisely the strict negative type metric matrices~\cite{Blumenthal1953theory}. We further characterize the maximum graph-variance problem and construct generalized resistive embeddings, extending the classical resistance geometry~\cite{Devriendt2022graph}.
\end{itemize}

\subsection{Outline}
\label{subsec:outline}
The remainder of this paper is organized as follows.
Section~2 reviews preliminaries on Laplacians, Kron reduction, effective resistance, graph variance, and metric geometry.
Section~3 develops the generalized Fiedler--Bapat identity and studies structural properties and resistance invariants under weight-balancing and Kron reduction. Section~4 extends the framework to signed undirected Laplacians and develops generalized resistance metrics, maximum graph-variance, and resistive embeddings.

\subsection{Notation}
\label{subsec:notation}

    Let $\bm{1}, \bm{0}$ be vectors of appropriate dimension whose components are all 1 and all 0, respectively. 
    Let $\bm{1}_i$ be the vector of appropriate dimension with 1 at position $i$ and $0$ at other positions. 
    For a matrix $A \in \mathbb{R}^{n \times n}$, when all components are non-negative, we write $A \geq O$. 
    We define the vector consisting of the diagonal elements of $A$ as $\operatorname{diagvec}(A) := [A_{11}, \dots, A_{nn}]^{\top} \in \mathbb{R}^n$. The eigenvalues of $A$ are written as  $\operatorname{sp}(A) = \{ \lambda_1(A), \dots, \lambda_n(A) \}$. 
    A symmetric matrix $A \in \mathbb{R}^{n \times n}$ is positive-semidefinite (PSD) if $\bm{x}^{\top}A\bm{x}\geq0$ for every $\bm{x}\in\mathbb{R}^n$, and we write $A \succeq O$. Given a finite set $X$, let $|X|$ be its cardinality, 
    and define the index set as $[n] = \{1,\ldots,n\}$ for $n \in \mathbb{N}$.
    For $\alpha,\beta \subset [n]$, let $A[\alpha ,\beta] \in \mathbb{R}^{|\alpha| \times |\beta|}$ denote the submatrix of $A$ obtained by the rows indexed by $\alpha$ and the columns indexed by $\beta$. When $\alpha = \beta$, we write this as $A[\alpha]$.
    Similarly, for a vector $\bm{v} \in \mathbb{R}^n$, 
    let $\bm{v}[\alpha] \in \mathbb{R}^{|\alpha|}$ denote the subvector obtained from the elements indexed by $\alpha$. 
    The \textit{Schur complement} of $A \in \mathbb{R}^{n \times n}$ with respect to a node set $\alpha$ is the $( |\alpha| \times |\alpha|)$-dimensional matrix, defined as
      \begin{equation*}\label{eq:Schur_complement}
        A/ \alpha^c := A / A[\alpha^c, \alpha^c] = A[\alpha, \alpha] - A[\alpha, \alpha^c] A[\alpha^c, \alpha^c]^{-1} A[\alpha^c, \alpha]
      \end{equation*}
    if $A[\alpha^c, \alpha^c]$ is invertible. 
    When we denote $\alpha^c = \{1, \dots, |\alpha^c|\}$, then $A / \alpha^c$ corresponds to the repeated Schur complements: 
    $A / \alpha^c = ((A/ \{ |\alpha^c| \}) / \cdots ) / \{1\}$ (Thm.I.3 of \cite{Zhang2005SchurComplement}).

    Let $\mathcal{G} = ([n], \mathcal{E}, \mathcal{A})$ be a graph without self-loops, 
    where $[n],\ \mathcal{E}$ and $\mathcal{A} \in \mathbb{R}^{n \times n}$ 
    denote the node set, the edge set, and the weighted adjacency matrix, respectively. 
    For off-diagonal elements, $\mathcal{A}_{ij} \neq 0$ means that there exists an edge from $i$ to $j$. 
    A graph is called \textit{unsigned} if all edge weights are non-negative; otherwise, it is called \textit{signed}. 
    The out-degree matrix is defined as 
    $\mathcal{D} := \text{diag} (\sum_{j=1}^n \mathcal{A}_{1j}, \dots, \sum_{j=1}^n \mathcal{A}_{nj})$ 
    where  $\diag(d_1, \dots, d_n)$ denotes the associated diagonal matrix.
    The Laplacian matrix is defined as $\mathcal{L} := \mathcal{D} - \mathcal{A}$. We denote a graph by $\mathcal{G}(\mathcal{L})$, in the sense that the structure of the graph is determined by the Laplacian.

    A directed graph $\mathcal{G}(\mathcal{L})$ is \textit{strongly connected (SC)} if there is a path between any pair of nodes, which is equivalent to the irreducibility of its Laplacian $\mathcal{L}$. Here, $X \in \mathbb{R}^{n \times n}$ is irreducible if for any $i \neq j$, there exists a sequence $k_1 = i, \dots, k_m = j$ such that $X_{k_1 k_2}, \dots, X_{k_{m-1} k_m}$ are all nonzero.

\section{Preliminaries}
\label{sec:preliminaries}
    To provide the necessary background for the generalizations developed in this paper, we review the Laplacian, Kron reduction, and effective resistance for unsigned directed graphs, as well as the maximum graph-variance problem and metric geometry in the classical unsigned undirected setting. These concepts have been extensively studied for unsigned undirected graphs \cite{Devriendt2022graph,Dorfler2013} and have more recently been extended to directed and signed graphs \cite{Fitch2019EffectiveResistance,Fontan2023pseudoinverse,Sugiyama2023graph,Young2016}. Unless otherwise stated, we assume that all graphs under consideration are connected and unsigned.

\subsection{Laplacian for graphs}
    We first recall the Laplacian matrix of a directed graph $\mathcal{G}(\mathcal{L}) = ([n], \mathcal{E}, \mathcal{A})$. 
    Its Laplacian satisfies $\mathcal{L}_{ij} \leq 0 \ (i, j \in [n],\ i \neq j)$, $\mathcal{L}_{ii} = -\sum_{j \neq i} \mathcal{L}_{ij} \ (\forall i \in [n])$, and $\mathcal{L}\bm{1} = \bm{0}$. 
    If there is a directed edge from $i$ to $j$, then $\mathcal{L}_{ij} = -\mathcal{A}_{ij} < 0$, and $\mathcal{L}_{ii}$ is equal to the weighted out-degree of node $i$. 
    As a special case, when the graph is undirected, $\mathcal{L}$ is symmetric. 
    In this case, the graph is connected if and only if $\mathcal{L}$ is PSD with its kernel being $\operatorname{span}(\bm{1})$ \cite{GodsilRoyle2001AGT}.
    
    In the remainder of this subsection, consider the case where $\mathcal{G}(\mathcal{L})$ is an SC-directed graph. 
    In this case, the eigenvalue 0 of $\mathcal{L}$ is simple and the other eigenvalues have positive real parts (Thm.I\hspace{-1.2pt}I of \cite{Altafini2019Investigating}). 
    Let the singular values of $\mathcal{L}$ be $s_i > 0 \ (\forall i \in [n-1])$ and $s_n = 0$. 
    Then, there exist orthogonal matrices 
    $U=[\bm{u}_1 \cdots \bm{u}_n]$ and $V=[\bm{v}_1 \cdots \bm{v}_n]$, 
    and $\Sigma=\operatorname{diag}(s_1, \dots, s_{n-1},0)$ 
    such that $\mathcal{L}=U\Sigma V^{\top}=\sum_{i=1}^{n-1} s_i\,\bm{u}_i\bm{v}_i^{\top}$ (SVD). 
    Combining $\mathcal{L}\bm{1}=\sum_{i=1}^{n-1} s_i (\bm{v}_i^{\top}\bm{1})\bm{u}_i = \bm{0}$ with the linear independence of $\bm{u}_1,\dots,\bm{u}_{n-1}$, 
    we have $\bm{v}_i^{\top} \bm{1} = 0 \ (\forall i \in [n-1])$, 
    meaning we can choose $\bm{v}_n = \bm{1}/\sqrt{n} $. 
    The pseudoinverse Laplacian is expressed as $\mathcal{L}^{\dagger} = \sum_{i=1}^{n-1} (1 / s_i) \bm{v}_i \bm{u}_i^{\top}$ and satisfies the following relations:

    \begin{equation}
        \label{eq:LLdagger_LdaggerL}
        \mathcal{L}\mathcal{L}^{\dagger} 
        = \sum_{i=1}^{n-1} \bm{u}_i \bm{u}_i^{\top} 
        = I_n - \bm{u}_n \bm{u}_n^{\top} 
        ,\quad
        \mathcal{L}^{\dagger} \mathcal{L}
        = \sum_{i=1}^{n-1} \bm{v}_i \bm{v}_i^{\top}
        = I_n - \bm{1} \bm{1}^{\top} / n.
    \end{equation}

    A directed graph $\mathcal{G}(\mathcal{L})$ is said to be \textit{weight-balanced (WB)} if 
    the in-degree and out-degree are equal at each node, which can be written as $\mathcal{L}\bm{1} = \mathcal{L}^{\top} \bm{1} = \bm{0}$. 
    Here we summarize the Laplacian properties of strongly connected weight-balanced (SCWB) directed graphs.
    The symmetrized Laplacian is defined as
    $\mathcal{L}_s := (\mathcal{L} + \mathcal{L}^{\top})/2$, 
    which is PSD with its kernel being $\operatorname{span}(\bm{1})$ (Cor.I of \cite{Altafini2019Investigating}). 
    Let $\mathcal{L}^{\dagger}_s$ be the symmetrized pseudoinverse Laplacian: $\mathcal{L}^{\dagger}_s := (\mathcal{L}^{\dagger} + (\mathcal{L}^{\dagger})^{\top})/2$.
    Here, note that $\mathcal{L}^{\dagger}_s \neq (\mathcal{L}_s)^{\dagger}$ (Rem.I\hspace{-1.2pt}V.6 of \cite{Fontan2023pseudoinverse}). 
    The following relations hold for any $\gamma \neq 0$ 
    (Lem.I\hspace{-1.2pt}V, Thm.I\hspace{-1.2pt}V of \cite{Fontan2021pseudoinverse}):
    \begin{equation}
        \label{eq:Ldaggers}
        \mathcal{L}^{\dagger} = (\mathcal{L} + \tfrac{\gamma }{n} \bm{1}\bm{1}^{\top})^{-1} - \tfrac{1}{n \gamma} \bm{1}\bm{1}^{\top},\quad
        \mathcal{L}^{\dagger}_s 
        = \left(\mathcal{L} + \tfrac{\gamma}{n} \bm{1}\bm{1}^{\top} \right)^{-1} \mathcal{L}_s \left( \mathcal{L} + \tfrac{\gamma}{n} \bm{1} \bm{1}^{\top} \right)^{-\top}.
    \end{equation} 
    Since $\mathcal{L}_s$ is PSD with its kernel being $\operatorname{span}(\bm{1})$, \cref{eq:Ldaggers} shows that $\mathcal{L}^{\dagger}_s$ is congruent to $\mathcal{L}_s$ through an invertible matrix. Moreover, $(\mathcal{L} + \tfrac{\gamma}{n} \bm{1}\bm{1}^{\top})^{\top}\bm{1} = \gamma\bm{1}$. Hence, $\mathcal{L}^{\dagger}_s$ is PSD with its kernel being $\operatorname{span}(\bm{1})$. 
    Since an SCWB-directed graph satisfies  $\ker(\mathcal{L}^{\top})=\operatorname{span}(\bm{1})$, we can choose $\bm{u}_n=\bm{1}/\sqrt{n}$. Hence, by \cref{eq:LLdagger_LdaggerL}, $\mathcal{L}\mathcal{L}^{\dagger} = I_n - \bm{1}\bm{1}^{\top} / n = \mathcal{L}^{\dagger} \mathcal{L}$.
    
    Let $(\mathcal{L}^{\dagger}_s)^{\dagger}$ be the pseudoinverse of $\mathcal{L}^{\dagger}_s$, which is symmetric PSD with its kernel being $\operatorname{span}(\bm{1})$. Thus, $(\mathcal{L}^{\dagger}_s)^{\dagger}$ can be thought of as the Laplacian of some connected undirected graph. Each node in $\mathcal{G}((\mathcal{L}^{\dagger}_s)^{\dagger})$ has positive degree, since $(\mathcal{L}^{\dagger}_s)^{\dagger}_{ii} = \bm{1}_i^{\top} (\mathcal{L}^{\dagger}_s)^{\dagger} \bm{1}_i > 0 \ (\forall i \in [n])$. However, its off-diagonal entries are not necessarily non-positive. Therefore, $\mathcal{G}((\mathcal{L}^{\dagger}_s)^{\dagger})$ is, in general, a signed undirected graph.
    \begin{example}
        \label{eq:exmaple_SCWB_directed_to_signed_undirected}
        Consider the following Laplacian of an SCWB-directed graph:
        \begin{equation*}
            \mathcal{L}
            =
            \begin{bmatrix}
              3 & -1 & -2 & 0 \\
              0 & 1 & 0 & -1 \\
              0 & 0 & 2 & -2 \\
              -3 & 0 & 0 & 3
            \end{bmatrix}
            \implies
            (\mathcal{L}^{\dagger}_s)^{\dagger}
            =
            \dfrac{1}{9}
            \begin{bmatrix}
              36 & -6 & -12 & -18 \\
              -6 & 10 & \mathbf{2} & -6 \\
              -12 & \mathbf{2} & 22 & -12 \\
              -18 & -6 & -12 & 36
            \end{bmatrix}
            \neq 
            \mathcal{L}_s.
        \end{equation*}
        The undirected edge $(2,3)$ in $\mathcal{G}((\mathcal{L}^{\dagger}_s)^{\dagger})$ has a negative weight, since the corresponding off-diagonal Laplacian entry is $(\mathcal{L}^{\dagger}_s)^{\dagger}_{23}=(\mathcal{L}^{\dagger}_s)^{\dagger}_{32} = 2/9 > 0$.
    \end{example}

    \cref{tab:Laplacian_connected_graphs} summarizes the properties of the Laplacian matrix when connected graphs are classified into four types based on the sign and direction of their edge weights. \cref{tab:SCWB_Laplacian_compare} summarizes the Laplacian matrices from unsigned SCWB-directed graphs.

    \begin{table}[!b]
        \centering
        \caption{Properties of the Laplacian matrix of a connected graph.}
        \label{tab:Laplacian_connected_graphs}
        \small
        \renewcommand{\arraystretch}{1.35}
        \begin{tabulary}{\textwidth}{c|c|c}
            \toprule[0.5pt] 
            & Unsigned & Signed \\ 
            \hline
            Undirected
            &
            $
            \mathcal{L}\succeq O,\ \ker(\mathcal{L})=\operatorname{span}(\bm{1}) \
            \cite{GodsilRoyle2001AGT}
            $
            &
            $
            \begin{gathered}
                \mathcal{L}\succeq O \text{ is not guaranteed}, \\
                \ker(\mathcal{L}) \supseteq \operatorname{span}(\bm{1}) \ \cite{Altafini2019Investigating,Fontan2023pseudoinverse}
            \end{gathered}
            $
            \\ \hline
            SC-directed
            &
            $
            \begin{gathered}
                \operatorname{Re}(\lambda)>0 \ (\forall \lambda\in\operatorname{sp}(\mathcal{L})\backslash \{0\} ), \\
                \ker(\mathcal{L})=\operatorname{span}(\bm{1}) \ \cite{Altafini2019Investigating}
            \end{gathered}
            $
            &
            $
            \begin{gathered}
                \operatorname{Re}(\lambda)>0 \ (\forall \lambda\in\operatorname{sp}(\mathcal{L}) \backslash \{0\})
                \text{ is not} \\
                \text{ guaranteed},\
                \ker(\mathcal{L}) \supseteq \operatorname{span}(\bm{1}) \ \cite{Altafini2019Investigating,Fontan2023pseudoinverse}
            \end{gathered}
            $
            \\
            \bottomrule[0.5pt] 
        \end{tabulary}
    \end{table}
    
    \begin{table}[!b]
        \centering
        \caption{A comparison of matrices from the Laplacian of an unsigned SCWB-directed graph.}
        \label{tab:SCWB_Laplacian_compare}
        \small
        \renewcommand{\arraystretch}{1}
        \begin{tabularx}{\textwidth}{cc|c|c|c|c}
            \toprule[0.5pt] 
            \multicolumn{2}{c|}{Type} & Kernel & Spectrum & Diagonal & Off-diagonal \\ 
            \hline
            $\mathcal{L}$ 
            & SCWB-directed & & Thm.I\hspace{-1.2pt}I of \cite{Altafini2019Investigating} & & \multirow{2}{*}{Non-positive} 
            \\ 
            \cline{4-4}
            $\mathcal{L}_s$ 
            & $(\mathcal{L} + \mathcal{L}^{\top})/2$ & & \multirow{3}{*}{\shortstack{$(\mathcal{L}^{\dagger}_s)^{\dagger} \succeq \mathcal{L}_s \succeq O$ \\ (see \cref{thm:Ldaggersdagger_decomposition}) }} & & 
            \\ 
            \cline{6-6}
            $\mathcal{L}^{\dagger}_s$ & $(\mathcal{L}^{\dagger} + (\mathcal{L}^{\dagger})^{\top})/2$ & \multirow{-2}{*}{$\operatorname{span}(\bm{1})$} & & \multirow{-2}{*}{Positive} & \multirow{2}{*}{Mixed} 
            \\ 
            $(\mathcal{L}^{\dagger}_s)^{\dagger}$ 
            & Pseudoinverse of $\mathcal{L}^{\dagger}_s$ & & & & 
            \\ 
            \bottomrule[0.5pt] 
        \end{tabularx}
    \end{table}

\subsection{Kron Reduction}
\label{subsec:KR}
\textit{Kron reduction} is a standard technique for reducing the size of a graph while preserving node connectivity, as well as effective resistance \cite{Dorfler2013, Kron1939TAN}. This method has been widely studied for circuit analysis and unsigned undirected graphs.
For an unsigned undirected graph $\mathcal{G}(\mathcal{L})$, the \textit{Kron-reduced Laplacian} $\mathcal{L}_{\text{red}} \in \mathbb{R}^{|\alpha| \times |\alpha|}$ with respect to the node set $\alpha \subset [n] \ (|\alpha| \geq 2)$ is defined as
\begin{equation}
    \mathcal{L}_{\text{red}} := \mathcal{L} / \alpha^c = \mathcal{L}[\alpha, \alpha] - \mathcal{L}[\alpha, \alpha^c] \mathcal{L}[\alpha^c, \alpha^c]^{-1}  \mathcal{L}[\alpha^c, \alpha].
\end{equation}
The matrix $\mathcal{L}_{\text{red}}$ preserves the Laplacian structure of an unsigned undirected graph: 
$\mathcal{L}_{\text{red}}$ is PSD with its kernel being $\operatorname{span}(\bm{1})$ and $(\mathcal{L}_{\text{red}})_{ij} \leq 0 \ (i \neq j)$ are satisfied. 
The graph $\mathcal{G}_{\text{red}} (:= \mathcal{G}(\mathcal{L}_{\text{red}}))$ is called the \textit{Kron-reduced graph}.

In recent research, Kron reduction has been extended to directed graphs and signed graphs \cite{Fontan2023pseudoinverse, Sugiyama2023graph, Young2016}. This paper mainly builds on the theory of Kron reduction and effective resistance for directed graphs developed by Sugiyama and Sato in \cite{Sugiyama2023graph}. They have designed a Kron reduction technique for unsigned directed graphs based on a node set property called ``reachability".
\begin{definition}[reachability, Def.I\hspace{-1.2pt}I\hspace{-1.2pt}I.2 of \cite{Sugiyama2023graph}]\label{def:reachable_subset}
    Let $\mathcal{G} = ([n], \mathcal{E}, \mathcal{A})$ be a directed graph. 
    The subset of nodes $\alpha \subset [n] \ (|\alpha| \geq 2)$ is said to be reachable if for any $i \in \alpha^c$, 
    there exists a node $j \in \alpha $ and a path in $\mathcal{G}$ from $i$ to $j$.
\end{definition}

When $\alpha$ is a reachable subset, 
$\mathcal{L}_{\text{red}}$ is well defined and the Laplacian structure of a directed graph, namely, $\mathcal{L} \bm{1} = \bm{0}$ and $\mathcal{L}_{ij} \leq 0 \ (i, j \in [n], i \neq j)$, 
are preserved after Kron reduction (Lem.I\hspace{-1.2pt}I\hspace{-1.2pt}I.3,4 of \cite{Sugiyama2023graph}).
Note also that Kron reduction keeps the WB property
(Thm.I\hspace{-1.2pt}I\hspace{-1.2pt}I.12 of \cite{Sugiyama2023graph}) 
and the SC property of the Laplacian. 

\subsection{Effective resistance}
\label{subsec:ER}
The \textit{effective resistance} has been widely studied in unsigned undirected graphs \cite{Devriendt2022graph, Dorfler2013}. 
Algebraically, the effective resistance between nodes is defined by the pseudoinverse Laplacian $\mathcal{L}^{\dagger}$:
\begin{equation}\label{eq:ER}
    R_{\mathcal{G}}(a, b) := (\bm{1}_a - \bm{1}_b)^{\top} \mathcal{L}^{\dagger} (\bm{1}_a - \bm{1}_b) \quad \forall a, b \in [n].
\end{equation}
Effective resistance is a metric and is invariant under Kron reduction (see \cite{Devriendt2022graph}).

Effective resistance has been extended to directed graphs in various ways.
In particular, Sugiyama and Sato in \cite{Sugiyama2023graph} proposed a definition based on random walks on directed graphs. 
The effective resistance $R_{\mathcal{G}}(a,b)$ from $a$ to $b$ is defined as
\begin{equation}\label{eq:def_ER}
    R_{\mathcal{G}}(a, b) := 1 / (\mathcal{L} / \{a, b\}^c)_{aa} \quad \text{if } \{a, b\} \subset [n] \text{ is reachable}
\end{equation}
(Def.I\hspace{-1.2pt}V.3 of \cite{Sugiyama2023graph}), and is asymmetric with respect to the node indices and is invariant under Kron reduction (Lem.I\hspace{-1.2pt}V.4 of \cite{Sugiyama2023graph}).
In an SC-directed graph, effective resistance is defined between any node pair based on subset reachability. When the graph is additionally weight-balanced, the effective resistance $R_{\mathcal{G}}(a, b)$ and the resistance matrix $\Omega := (R_{\mathcal{G}}(a, b))_{a, b \in [n]}$ are expressed as
\begin{equation}
    \label{eq:ER_SCWB}
    \begin{gathered}
    R_{\mathcal{G}}(a,b) = (\bm{1}_a - \bm{1}_b)^{\top} \mathcal{L}_s^{\dagger}(\bm{1}_a - \bm{1}_b) = (\mathcal{L}^{\dagger}_s)_{aa} + (\mathcal{L}^{\dagger}_s)_{bb} -2 (\mathcal{L}^{\dagger}_s)_{ab},\\
    \Omega = \bm{1} \bm{\zeta}^{\top} + \bm{\zeta} \bm{1}^{\top} -2\mathcal{L}^{\dagger}_s \quad \text{ where } \bm{\zeta} = \operatorname{diagvec}(\mathcal{L}^{\dagger}_s)
    \end{gathered}
\end{equation}
(Thm.I\hspace{-1.2pt}V.9 of \cite{Sugiyama2023graph}). Moreover, $R_\mathcal{G}$ is the metric between nodes (Thm.I\hspace{-1.2pt}V.11 of \cite{Sugiyama2023graph}). 

For an SC-directed graph $\mathcal{G}(\mathcal{L})$, there exists a matrix $M = \operatorname{diag}(m_1, ..., m_n) > O$ 
such that $\mathcal{L}_{\text{WB}} := M \mathcal{L}$ is the Laplacian of an SCWB-directed graph (Thm.I\hspace{-1.2pt}I of \cite{Altafini2019Investigating}). The effective resistance in $\mathcal{G}(\mathcal{L})$ and that in $\mathcal{G}_{\text{WB}} (:= \mathcal{G}(\mathcal{L}_{\text{WB}}))$ have the following relationship (Thm.I\hspace{-1.2pt}V.16 of \cite{Sugiyama2023graph}):
\begin{equation}
    \label{eq:ER_onlySC}
    R_{\mathcal{G}}(a, b) = m_a (\bm{1}_a - \bm{1}_b)^{\top} (\mathcal{L}_{\text{WB}})^{\dagger}_s (\bm{1}_a - \bm{1}_b) = m_a R_{\mathcal{G}_{\text{WB}}}(a, b)
    \quad \forall a, b \in [n].
\end{equation}

\subsection{Maximum variance problem}
\label{sec:maximum_variance_problem}
Let $\mathcal{G} = ([n], \mathcal{E}, \mathcal{A})$ be a graph with a metric $d$ between its nodes. 
The \textit{graph-variance} is defined as $\operatorname{var}_d (\bm{f}) := \bm{f}^{\top} D \bm{f}/2$ for any $\bm{f} \in \Delta_n$. 
Here, $D := (d(i, j))_{i, j \in [n]}$ is a distance matrix 
and $\Delta_n := \{\bm{f} \geq \bm{0} \mid \bm{f}^{\top} \bm{1} = 1 \}$ is the probability simplex on $\mathbb{R}^n$. 
Graph-variance corresponds to the mean squared distance between two nodes chosen according to the probability distribution $\bm{f} = (f_i)_{i \in [n]}$ on the nodes (Def.V.1 of \cite{Devriendt2022graph}).

The \textit{maximum variance problem} (see Chap.V of \cite{Devriendt2022graph}) is defined as 
\begin{equation}\label{eq:maxmium_variance_problem}
    \begin{aligned}
    & \underset{\bm{f} \in \Delta_n}{\text{maximize}} && \operatorname{var}_d (\bm{f}).
\end{aligned}
\end{equation}
Since $\Delta_n$ is compact and $\operatorname{var}_d$ is continuous,
the maximum is attained at some $\bm{f}^* \in \Delta_n$, called the \textit{maximum variance distribution}.
A necessary condition for $\bm{f} \in \Delta_n$ to be the maximum variance distribution is given in Prop.V.13 of \cite{Devriendt2022graph}:
\begin{equation}
    \label{eq:necesaary_condition_2}
    D[\mathcal{V}, \mathcal{V}] \bm{f}[\mathcal{V}] = (\bm{1}^{\top} D[\mathcal{V}, \mathcal{V}]^{\dagger} \bm{1})^{-1} \bm{1}, 
\end{equation}
where $\mathcal{V} := \{ i \mid f_i > 0 \}$ is the support of $\bm{f}$, termed the \textit{maximum variance support}. 

\subsection{Classification of metric spaces}
\label{subsec:classification}
Here, we summarize and classify metric spaces.
A metric space $(X, d)$ is \textit{negative type} 
if its distance matrix satisfies
$\bm{f}^{\top} D \bm{f} \leq 0 \text{ for all } \bm{f} \in \operatorname{span}(\bm{1})^{\perp}$. 
If the inequality is strict for all
$\bm{f} \in \operatorname{span}(\bm{1})^{\perp} \setminus \{ \bm{0} \}$, 
it is called \textit{strict negative type} (Def.I\hspace{-1.2pt}V.8 of \cite{Devriendt2022graph}). 
It is known that negative type metric spaces have an embedding 
into the Euclidean space $\mathbb{R}^{n-1}$ (\cite{Blumenthal1953theory}, Prop.I\hspace{-1.2pt}V.11 of \cite{Devriendt2022graph}). 
Here, 
$B = [\bm{b}_1 \ \cdots \ \bm{b}_n] \in \mathbb{R}^{(n-1) \times n}$ and 
$E = (\| \bm{b}_i - \bm{b}_j \|^2) \in \mathbb{R}^{n \times n}$ 
are matrices consisting of the coordinates of $n$ points and the squared Euclidean distance matrix, respectively. 
The \textit{centered Gram matrix} $G$ is defined as 
$G := (I_n - \bm{1}\bm{1}^{\top}/n) B^{\top} B (I_n - \bm{1}\bm{1}^{\top}/n)$, 
and the following relations hold:
\begin{equation}
    \label{eq:gram_euclid}
    \begin{gathered}
        G = -\frac{1}{2} (I_n - \bm{1}\bm{1}^{\top}/n) E (I_n - \bm{1}\bm{1}^{\top} /n), \\
        E = \bm{1} \bm{\xi}^{\top} + \bm{\xi} \bm{1}^{\top} - 2G \ \text{ where } \  \bm{\xi} := \operatorname{diagvec}(G).
    \end{gathered}
\end{equation}
$G$ is PSD with its kernel containing (being) $\operatorname{span}(\bm{1})$ if and only if $G$ is a centered Gram matrix (of simplex) (Prop.I\hspace{-1.2pt}V.13, 15 of \cite{Devriendt2022graph}). 
The effective resistance in an unsigned undirected graph becomes a metric, known as \textit{resistance metric}. 
A brief summary and the inclusion relations of metric spaces are shown in \cref{fig:classification}.

\begin{figure}[t]
  \centering
  \begin{minipage}[t]{0.40\textwidth}
    \vspace{0pt} 
    \centering
    \begin{tikzpicture}[scale=0.62] 
      \begin{scope}[every node/.style={font=\scriptsize}]
        \fill[white!30, opacity=0.3] (0,0) ellipse (4cm and 2.5cm);
        \draw[black, thick] (0,0) ellipse (4cm and 2.5cm);
        \node[align=center] at (0,1.8) {\textbf{Negative type}};

        \fill[white!30, opacity=0.3] (0,-0.5) ellipse (2.7cm and 1.7cm);
        \draw[black, thick] (0,-0.5) ellipse (2.7cm and 1.7cm);
        \node[align=center] at (0,0.5) {\textbf{Strict} \\ \textbf{negative type}};

        \fill[white!30, opacity=0.3] (0,-1.0) ellipse (1.5cm and 1.0cm);
        \draw[black, thick] (0,-1.0) ellipse (1.5cm and 1.0cm);
        \node[align=center] at (0,-1.0) {\textbf{Resistance} \\ \textbf{metric space}};
      \end{scope}
    \end{tikzpicture}
  \end{minipage}%
  \hspace{0.02\textwidth}%
  \begin{minipage}[t]{0.58\textwidth}
    \vspace{0pt} 
    \raggedright
    \small
    
    \textbf{Characterizations}

    \medskip
    \textbf{Negative type:}
    $\forall \bm{f} \in \operatorname{span}(\bm{1})^{\perp},\ \bm{f}^{\top} D \bm{f} \leq 0$.
    \\
    $\exists \{ \varphi(\bm{1}_1), \dots, \varphi(\bm{1}_n)\}$: embedded points in $\mathbb{R}^{n-1}$.

    \medskip
    \textbf{Strict negative type:}
    $\forall \bm{f} \in \operatorname{span}(\bm{1})^{\perp}\backslash\{\bm{0}\},\ \bm{f}^{\top} D \bm{f} < 0$.
    \\
    Embedded $n$ points form a simplex in $\mathbb{R}^{n-1}$.

    \medskip
    \textbf{Resistance metric:}
    $D = \Omega,\ G = \mathcal{L}^{\dagger},\ G^{\dagger} = \mathcal{L}$.
    \\
    Embedded $n$ points form a hyperacute simplex \cite{Devriendt2022graph}.
  \end{minipage}

  \caption{The hierarchical relations among three classes of metric spaces: Negative type metric $\supsetneq$ Strict negative type metric $\supsetneq$ Resistance metric space. For a detailed explanation, see \cite{Devriendt2022graph, Fiedler2011_2}.}
  \label{fig:classification}
\end{figure}

\section{Main results}
\label{sec:main_results}
In this section, we establish a generalized Fiedler--Bapat identity for SCWB-directed graphs.
We then compare the symmetrized Laplacian with the associated signed undirected Laplacian $(\mathcal{L}^{\dagger}_s)^{\dagger}$, identify the positive-semidefinite correction induced by graph directionality, and derive a corresponding effective-resistance comparison. 
We next prove that the map $\mathcal{L} \mapsto (\mathcal{L}^{\dagger}_s)^{\dagger}$ and weight-balancing both commute with Kron reduction.
For SC-directed graphs satisfying $\bm{1}^{\top}\Omega^{-1}\bm{1}>0$, we define the resistance curvature and resistance radius through a weight-balanced representation. Finally, we establish the monotonicity of the resistance radius under Kron reduction in the SCWB-directed setting.

\subsection{Generalized Fiedler--Bapat identity}
\label{subsec:FBI}

The FBI is an inverse matrix identity that relates the Laplacian to the resistance matrix for an unsigned undirected graph. It was derived by M.~Fiedler in the context of simplex geometry \cite{Fiedler2011} and is thoroughly explained in Chap.I\hspace{-1.2pt}I\hspace{-1.2pt}I of \cite{Devriendt2022graph}. Here, we generalize the FBI and propose a matrix identity in the SC-directed setting. First, we present the generalized FBI in the SCWB-directed setting. We focus on three points:
\begin{itemize}
    \item 
    \textbf{Strongly connected:} 
    When $\mathcal{G}(\mathcal{L})$ is SC-directed, 
    the effective resistance is defined between any node pair; thus the resistance matrix $\Omega$ is well-defined.
    \item 
    \textbf{Weight-balanced:}
    In an SCWB-directed graph, 
    the effective resistance is a metric that can be expressed by $\mathcal{L}^{\dagger}_s$ \cref{eq:ER_SCWB}. 
    The matrix $\mathcal{L}^{\dagger}_s$ is PSD with its kernel being $\operatorname{span}(\bm{1})$
    (Thm.I\hspace{-1.2pt}V.4 of \cite{Fontan2023pseudoinverse}) and $\Omega$ is symmetric. 
    \item 
    \textbf{Comparing the expressions of the effective resistance:}
    We compare two formulas for the effective resistance: \cref{eq:ER} for undirected graphs and \cref{eq:ER_SCWB} for SCWB-directed graphs.  
    This suggests that $(\mathcal{L}^{\dagger}_s)^{\dagger}$ plays the role as a Laplacian matrix in the right-hand side of the FBI for SCWB-directed setting.
\end{itemize}

\begin{theorem}[FBI for strongly connected weight-balanced directed graphs]\label{thm:FBI_generalized}
  Let $\mathcal{L}$ be the Laplacian corresponding to an SCWB-directed graph, and let $\bm{\zeta} := \operatorname{diagvec}(\mathcal{L}^{\dagger}_s)$.  
  The resistance matrix $\Omega$ and Laplacian matrix $\mathcal{L}$ satisfy
  \begin{equation}\label{eq:FBI_generalized}
    \begin{bmatrix}
      0 &\bm{1}^{\top} \\
      \bm{1} & \Omega
    \end{bmatrix}^{-1}
    =
    -\dfrac{1}{2}
    \begin{bmatrix}
      4\sigma^2 & -2\bm{p}^{\top} \\ 
      -2\bm{p} & (\mathcal{L}^{\dagger}_s)^{\dagger}
    \end{bmatrix}
    \ \text{where} \
    \begin{cases}
      \bm{p} = \dfrac{1}{2} (\mathcal{L}^{\dagger}_s)^{\dagger} \bm{\zeta} + \dfrac{1}{n} \bm{1}, \\
      \sigma^2 = \dfrac{1}{4} \bm{\zeta}^{\top} (\mathcal{L}^{\dagger}_s)^{\dagger} \bm{\zeta} + \dfrac{1}{n} \bm{1}^{\top} \bm{\zeta}.
    \end{cases}
  \end{equation}
\end{theorem}

\begin{proof}
  First, we show that $\Omega$ is invertible. From $\cref{eq:ER_SCWB}$ and the fact that $\mathcal{L}^{\dagger}_s$ is PSD with its kernel being $\operatorname{span}(\bm{1})$ (Thm.I\hspace{-1.2pt}V.4 of \cite{Fontan2023pseudoinverse}), 
  we have that $\bm{1}^{\top} \Omega \bm{1} > 0$ and 
  $\bm{f}^{\top} \Omega \bm{f} = -2\bm{f}^{\top} \mathcal{L}^{\dagger}_s \bm{f} < 0$ for all $\bm{f} \in \operatorname{span}(\bm{1})^{\perp} \backslash \{ \bm{0}\}$. 
  Choose an orthogonal matrix $P= [\bm{1}/\sqrt{n}  \mid  U ]$ 
  where the columns of $U \in \mathbb{R}^{n \times (n-1)}$ form an orthonormal basis of $\operatorname{span}(\mathbf 1)^{\perp}$. 
  Since $\Omega$ is a symmetric matrix, 
  \begin{equation*}
    \exists a \in \mathbb{R},\ \exists \bm{b} \in \mathbb{R}^{n-1},\ \exists C \in \mathbb{R}^{(n-1) \times (n-1)} \text{ s.t. }
    P^{\top}\Omega P = \begin{bmatrix} a & \bm{b}^{\top} \\  \bm{b} & C \end{bmatrix}  \eqqcolon {\Omega}^{\prime}.
  \end{equation*}
  For the blocks of $\Omega^{\prime}$, note that $ a = \bm{1}^{\top} \Omega \bm{1} / n > 0$ 
  and 
  $\bm{y}^{\top} C \bm{y} = (U\bm{y})^{\top} \Omega (U \bm{y}) < 0$ for all $\bm{y} \in \mathbb{R}^{n-1} \backslash \{ \bm{0}\}$, 
  that is, $C$ is negative-definite. Moreover, 
  \begin{equation*} 
    S^{\top} {\Omega}^{\prime}S 
    = 
    \begin{bmatrix} 
        a - \bm{b}^{\top}C^{-1} \bm{b} & \bm{0}^{\top} \\ 
        \bm{0} & C 
    \end{bmatrix}
    =: \Omega^{\prime\prime}
    \ \text{ where } \
    S := 
    \begin{bmatrix} 
        1 & \bm{0}^{\top} \\ 
        -C^{-1} \bm{b} & I_{n-1} 
    \end{bmatrix}
  \end{equation*}
  holds. Note that $a - \bm{b}^{\top} C^{-1} \bm{b} \geq a > 0$. 
  Sylvester's law of inertia shows that the matrix $\Omega$ has the same inertia as the block-diagonal matrix $\Omega^{\prime\prime}$ and is therefore invertible.

  Next, we prove that the block matrix
  $
    \bigl[ \begin{smallmatrix}
     0 & \bm{1}^{\top} \\
     \bm{1} & \Omega 
    \end{smallmatrix} \bigr]
  $
  is invertible. Suppose that this matrix is singular. Then, 
  \begin{equation*}
    \label{eq:FBI_generalized_1}
     \exists
    \begin{bmatrix}
      -c \\
      \bm{f}
    \end{bmatrix}
    \neq \bm{0}
    \text{ s.t. }
    \begin{bmatrix}
      0 & \bm{1}^{\top} \\
      \bm{1} & \Omega
    \end{bmatrix}
    \begin{bmatrix}
      -c \\
      \bm{f}
    \end{bmatrix}
    =
    \begin{bmatrix}
      0 \\
      \bm{0}
    \end{bmatrix}
    \Longleftrightarrow
    \begin{cases}
      \bm{1}^{\top} \bm{f} = 0, \\
      \Omega \bm{f} = c \bm{1}.
    \end{cases}
  \end{equation*}
  Setting $c = 0$ yields $\bm{f} = \bm{0}$ by the invertibility of $\Omega$, a contradiction.
  If instead $c \neq 0$, then $\bm{f} \in \operatorname{span}(\bm{1})^{\perp} \backslash \{ \bm{0} \}$. This leads to
  \begin{equation*}
    0 = c \bm{f}^{\top} \bm{1} = \bm{f}^{\top} (c \bm{1}) = \bm{f}^{\top} \Omega \bm{f} 
    \underset{\text{\cref{eq:ER_SCWB}}}{=} -2 \bm{f}^{\top} \mathcal{L}^{\dagger}_s \bm{f} < 0, 
  \end{equation*}
  which is again a contradiction. Therefore, the matrix 
  $
    \bigl[ \begin{smallmatrix}
     0 & \bm{1}^{\top} \\
     \bm{1} & \Omega 
    \end{smallmatrix} \bigr]
  $
  is invertible. 

  Consider the matrix on the right-hand side of \cref{eq:FBI_generalized}. 
  Since the matrix on the left-hand side of \cref{eq:FBI_generalized} is invertible, 
  there exist $\bm{p} \in \mathbb{R}^n,\ \sigma^2 \in \mathbb{R}$ and $ \mathcal{X} \in \mathbb{R}^{n \times n}$ such that 
  \begin{equation}\label{eq:FBI_generalized_2}
    \begin{bmatrix}
      0 & \bm{1}^{\top} \\
      \bm{1} & \Omega
    \end{bmatrix}
    \begin{bmatrix}
      4\sigma^2 & -2 \bm{p}^{\top} \\
      -2 \bm{p} & \mathcal{X}
    \end{bmatrix}
    =
    -2I_{n+1}
    \Longleftrightarrow 
    \begin{cases}
      \bm{1}^{\top} \bm{p} = 1, \
      2\sigma^2 \bm{1} = \Omega \bm{p}, \\
      \Omega \mathcal{X} = -2 I_n + 2 \bm{1} \bm{p}^{\top}, \
      \mathcal{X} \bm{1} = \bm{0}.
    \end{cases}
  \end{equation}
    Combining $\bm{1}^{\top} \bm{p} = 1$, $2\sigma^2 \bm{1} = \Omega \bm{p}$ and the relation \cref{eq:ER_SCWB}, we obtain
    \begin{equation*}
        2\sigma^2 \bm{1} = (\bm{1} \bm{\zeta}^{\top} + \bm{\zeta} \bm{1}^{\top} - 2 \mathcal{L}^{\dagger}_s )\bm{p} \notag 
        \implies
        (2 \sigma^2 - \bm{\zeta}^{\top} \bm{p} - \bm{\zeta}^{\top} \bm{1} / n) \bm{1}
        = 
        (I_n - \bm{1} \bm{1}^{\top} / n) \bm{\zeta} - 2 \mathcal{L}^{\dagger}_s \bm{p}.
    \end{equation*}
    The left-hand side above is contained in $\operatorname{span}(\bm{1})$, while the right-hand side above is contained in $\operatorname{span}(\bm{1})^{\perp}$. This implies that $2\sigma^2 = (\bm{p} + \bm{1}/n)^{\top} \bm{\zeta}$ and $2\mathcal{L}^{\dagger}_s \bm{p} = (I_n - \bm{1}\bm{1}^{\top} / n) \bm{\zeta}$ hold. Left-multiplying the latter equation by the matrix $(\mathcal{L}^{\dagger}_s)^{\dagger}$ yields the following concrete expressions for the vector $\bm{p}$ and the scalar $\sigma^2$:
    \begin{equation}\label{eq:concrete_p_sigma2}
        \bm{p} = \dfrac{1}{2}(\mathcal{L}^{\dagger}_s)^{\dagger} \bm{\zeta} + \frac{1}{n} \bm{1},\quad
        \sigma^2 = \dfrac{1}{4} \bm{\zeta}^{\top} (\mathcal{L}^{\dagger}_s)^{\dagger} \bm{\zeta} + \dfrac{1}{n} \bm{1}^{\top} \bm{\zeta}.
    \end{equation}
    Given that $\mathcal{L}^{\dagger}_s$ is PSD with its kernel being $\operatorname{span}(\bm{1})$ (Thm.I\hspace{-1.2pt}V.4 of \cite{Fontan2023pseudoinverse}), we have $\sigma^2 > 0$. 

    Finally, we show that the matrix $\mathcal{X}$ is $(\mathcal{L}^{\dagger}_s)^{\dagger}$. 
    Since $\Omega \mathcal{X} = -2 I_n + 2 \bm{1} \bm{p}^{\top}$ and $\mathcal{X}\bm{1} = \bm{0}$ from \cref{eq:FBI_generalized_2}, and since $\Omega$ is invertible, the matrix $\mathcal{X}$ is uniquely determined by these relations. 
    Therefore, it suffices to verify that $(\mathcal{L}^{\dagger}_s)^{\dagger}$ satisfies the same equations.
    Using \cref{eq:ER_SCWB}, we have $\Omega = \bm{1} \bm{\zeta}^{\top} + \bm{\zeta} \bm{1}^{\top} - 2 \mathcal{L}^{\dagger}_s$. 
    Moreover, since $\mathcal{L}^{\dagger}_s$ is symmetric PSD with its kernel being $\operatorname{span}(\bm{1})$, its pseudoinverse $(\mathcal{L}^{\dagger}_s)^{\dagger}$ satisfies $(\mathcal{L}^{\dagger}_s)^{\dagger}\bm{1}=\bm{0}$ and $\mathcal{L}^{\dagger}_s(\mathcal{L}^{\dagger}_s)^{\dagger} = I_n-\bm{1} \bm{1}^{\top}/n$. 
    Hence,
    $\Omega (\mathcal{L}^{\dagger}_s)^{\dagger}
    =
    \bigl(\bm{1} \bm{\zeta}^{\top} + \bm{\zeta} \bm{1}^{\top} - 2 \mathcal{L}^{\dagger}_s \bigr)(\mathcal{L}^{\dagger}_s)^{\dagger}=
    \bm{1}\bigl((\mathcal{L}^{\dagger}_s)^{\dagger}\bm{\zeta}\bigr)^{\top}
    -2\bigl(I_n-\bm{1}\bm{1}^{\top}/n\bigr)$.
    By \eqref{eq:concrete_p_sigma2}, we obtain $\Omega (\mathcal{L}^{\dagger}_s)^{\dagger} = -2I_n+2\bm{1}\bm{p}^{\top}$. 
\end{proof}

The FBI involves two resistance-based invariants \cite{Devriendt2022graph}: the node-wise \textit{resistance curvature} $\bm{p}=(p_i)_{i\in[n]}$ and the scalar \textit{resistance radius} $\sigma^2$, which measures the global graph size. These quantities govern the maximum variance problem and resistive embeddings. For detailed treatments in the unsigned undirected setting, see \cite{Devriendt2022graph, Dawkins2024NodeRC, Devriendt2024resistancedistance}.

From the proof of \cref{thm:FBI_generalized}, we obtain the following relations. 
\begin{corollary}\label{cor:FBI_generalized}
  The following relations among $\mathcal{L}, \Omega, \bm{p}$ and $\sigma^2$ for an SCWB-directed graph hold:
    \begin{align}
      & \label{eq:FBI_generalized_relations_1}
      \bm{1}^{\top} \bm{p} = 1, \quad \Omega \bm{p} = 2 \sigma^2 \bm{1},\ \quad \Omega (\mathcal{L}^{\dagger}_s)^{\dagger} = -2 I_n + 2 \bm{1} \bm{p}^{\top}. \\
      & \label{eq:FBI_generalized_relations_2}
      \Omega^{-1} = -\frac{1}{2}(\mathcal{L}^{\dagger}_s)^{\dagger} + \frac{ \bm{p} \bm{p}^{\top}}{2 \sigma^2},\ \quad (\mathcal{L}^{\dagger}_s)^{\dagger} = -2 \Omega^{-1} + 2 \frac{ \Omega^{-1} \bm{1} \bm{1}^{\top} \Omega^{-1}}{\bm{1}^{\top} \Omega^{-1}\bm{1}}. \\
      & \label{eq:FBI_generalized_relations_3}
      \bm{p} = \frac{\Omega^{-1} \bm{1}}{\bm{1}^{\top} \Omega^{-1} \bm{1}},\ \quad \sigma^2 = \frac{1}{2}(\bm{1}^{\top} \Omega^{-1} \bm{1})^{-1},\ \quad \sigma^2 = \frac{1}{2} \bm{p}^{\top} \Omega \bm{p}.
    \end{align}
\end{corollary}
This corollary is similar to Cor.I\hspace{-1.2pt}I\hspace{-1.2pt}I.7,8,10 of \cite{Devriendt2022graph} 
and can be proven from \cref{thm:FBI_generalized}. 

Consider an SC-directed (not necessarily WB) graph $\mathcal{G}(\mathcal{L})$. With a matrix $M = \operatorname{diag}(m_1, \dots, m_n) > O$, we can obtain the Laplacian matrix for an SCWB-directed graph (Thm.I\hspace{-1.2pt}I of \cite{Altafini2019Investigating}). 
The matrix $\mathcal{L}_{\text{WB}} := M \mathcal{L}$ satisfies $\mathcal{L}_{\text{WB}} \bm{1} = \mathcal{L}_{\text{WB}}^{\top} \bm{1} = \bm{0}$. 
We can apply the FBI shown in \cref{thm:FBI_generalized} to the graph $\mathcal{G}(\mathcal{L}_{\text{WB}})$ as follows:
\begin{equation}\label{eq:FBI_simple_SC}
    \begin{bmatrix}
      0 & \bm{1}^{\top} \\
      \bm{1} & \Omega_{\text{WB}}
    \end{bmatrix}^{-1}
    =
    -\frac{1}{2}
    \begin{bmatrix}
      4 \sigma_{\text{WB}}^2 & -2 \bm{p}_{\text{WB}}^{\top} \\
      -2 \bm{p}_{\text{WB}} & ((\mathcal{L}_{\text{WB}})^{\dagger}_s)^{\dagger}
    \end{bmatrix} 
    \text{ where }
    \begin{cases}
        \mathcal{L}_{\text{WB}} = M\mathcal{L}, \\
        \Omega_{\text{WB}} = M^{-1} \Omega & (\because \ref{eq:ER_onlySC})
    \end{cases}
\end{equation}
Here, $\Omega_{\text{WB}},\ \bm{p}_{\text{WB}}$, and $\sigma_{\text{WB}}^2$ are the resistance matrix, 
the resistance curvature, and the resistance radius, respectively, in the weight-balanced graph $\mathcal{G}(\mathcal{L}_{\text{WB}})$. 

\subsection{Comparison of two undirected Laplacians}
\label{subsec:comparison_of_two_undirected_Laplacians}

In this subsection, we examine the relationship between an unsigned SCWB-directed Laplacian $\mathcal{L}$ and its associated signed undirected Laplacian $(\mathcal{L}^{\dagger}_s)^{\dagger}$. Throughout this subsection, we write $\mathcal{L}_s := (\mathcal{L}+\mathcal{L}^{\top})/2$ and $\mathcal{K} := (\mathcal{L}-\mathcal{L}^{\top})/2$.

We first record a representation of the relevant pseudoinverses obtained by restricting the Laplacians to $\operatorname{span}(\bm{1})^{\perp}$.

\begin{lemma}
    \label{lem:lem_for_relationship_between_L_and_Ldaggersdagger}
    Let $\mathcal{L}$ be an unsigned SCWB-directed Laplacian. Given a matrix $U \in \mathbb{R}^{n \times (n-1)}$ whose orthonormal columns span $\bm{1}^{\perp}$, define $\bar{\mathcal{L}} := U^{\top}\mathcal{L}U$ and $\bar{\mathcal{L}}_s := U^{\top} \mathcal{L}_s U$. Then, $\bar{\mathcal{L}}$ and $\bar{\mathcal{L}}_s$ are invertible, and 
    \begin{equation}
        \label{eq:eq_for_relationship_between_L_and_Ldaggersdagger}
        \mathcal{L}^{\dagger} = U \bar{\mathcal{L}}^{-1}U^{\top},\quad (\mathcal{L}_s)^{\dagger} = U (\bar{\mathcal{L}}_s)^{-1} U^{\top}.
    \end{equation}
\end{lemma}
\begin{proof}
    We note that $\mathcal{L}_s$ is PSD with its kernel being $\operatorname{span}(\bm{1})$. Moreover, since $\mathcal{L}\bm{1} = \mathcal{L}^{\top} \bm{1} = \bm{0}$ and $U^{\top}U = I_{n-1},\ UU^{\top} = I_n - \bm{1}\bm{1}^{\top}/n$, we obtain $\mathcal{L} = U \bar{\mathcal{L}}U^{\top}$ and $\mathcal{L}_s = U \bar{\mathcal{L}}_sU^{\top}$. 
    
    Suppose $\bar{\mathcal{L}}\bm{y} = U^{\top} \mathcal{L} U \bm{y} = \bm{0}$ for some $\bm{y} \in \mathbb{R}^{n-1}$. Left-multiplying by $U$ gives $\mathcal{L}U\bm{y} = \bm{0}$, so $U\bm{y} \in \operatorname{span}(\bm{1})$. Since $U\bm{y} \in \operatorname{span}(\bm{1})^{\perp}$ by construction, it follows that $U\bm{y} \in \operatorname{span}(\bm{1}) \cap \operatorname{span}(\bm{1})^{\perp} = \{\bm{0}\}$, which yields $\bm{y} = \bm{0}$. Hence, $\bar{\mathcal{L}}$ is invertible. Furthermore, the invertibility of $\bar{\mathcal{L}}_s$ follows from the fact that $\bar{\mathcal{L}}_s$ is positive-definite; this is derived from $\bm{y}^{\top} \bar{\mathcal{L}}_s \bm{y} = (U \bm{y})^{\top} \mathcal{L}_s (U \bm{y}) > 0 \ (\forall\bm{y} \neq \bm{0})$.

    Next, it can be verified that $U \bar{\mathcal{L}}^{-1} U^{\top} (=: \mathcal{X})$ is the pseudoinverse of $\mathcal{L}$ as follows.
    \begin{equation*}
        \begin{gathered}
            \mathcal{L}\mathcal{X} = (\mathcal{L}\mathcal{X})^{\top} = \mathcal{L}\mathcal{L}^{\dagger}
            \underset{\text{\cref{eq:LLdagger_LdaggerL}}}{=} 
            I_n - \bm{1}\bm{1}^{\top}/n, 
            \quad
            \mathcal{X}\mathcal{L} = (\mathcal{X}\mathcal{L})^{\top} = \mathcal{L}^{\dagger}\mathcal{L}
            \underset{\text{\cref{eq:LLdagger_LdaggerL}}}{=} I_n - \bm{1}\bm{1}^{\top}/n, \\
            \mathcal{L}\mathcal{X}\mathcal{L} = ( I_n - \bm{1}\bm{1}^{\top}/n )\mathcal{L} = \mathcal{L}, 
            \quad
            \mathcal{X}\mathcal{L}\mathcal{X} = \mathcal{X} ( I_n - \bm{1}\bm{1}^{\top}/n ) = \mathcal{X}.
        \end{gathered}
    \end{equation*}
    Similarly, $U (\bar{\mathcal{L}}_s)^{-1} U^{\top}$ is the pseudoinverse of $\mathcal{L}_s$.
\end{proof}

\cref{lem:lem_for_relationship_between_L_and_Ldaggersdagger} clarifies the precise contribution of the asymmetric part $\mathcal{K}$ of $\mathcal{L}$.

\begin{theorem}[Decomposition of the associated signed undirected Laplacian]
    \label{thm:Ldaggersdagger_decomposition}
    Let $\mathcal{L}$ be an unsigned SCWB-directed Laplacian. Then, 
    \begin{equation}
        \label{eq:Ldaggersdagger_decomposition}
        (\mathcal{L}^{\dagger}_s)^{\dagger} = \mathcal{L}(\mathcal{L}_s)^{\dagger}\mathcal{L}^{\top} = \mathcal{L}_s + \mathcal{K}(\mathcal{L}_s)^{\dagger}\mathcal{K}^{\top} \succeq \mathcal{L}_s.
    \end{equation}
    Moreover, the equality in $(\mathcal{L}^{\dagger}_s)^{\dagger} \succeq \mathcal{L}_s$ holds if and only if $\mathcal{L}$ is symmetric.
\end{theorem}
\begin{proof}
    By definition, $\bar{\mathcal{L}}_s = U^{\top}\mathcal{L}_s U = (\bar{\mathcal{L}} + \bar{\mathcal{L}}^{\top})/2$. Since $\bar{\mathcal{L}}$ is invertible and $\bar{\mathcal{L}}_s$ is positive-definite, 
    we have that the matrix $(\bar{\mathcal{L}}^{-1}+\bar{\mathcal{L}}^{-\top})/2 = \bar{\mathcal{L}}^{-\top} \bar{\mathcal{L}}_s \bar{\mathcal{L}}^{-1}$ is also positive-definite. 
    It follows from \cref{lem:lem_for_relationship_between_L_and_Ldaggersdagger} that
    \begin{equation*}
        \mathcal{L}^{\dagger}_s = U \biggl( \frac{\bar{\mathcal{L}}^{-1} + \bar{\mathcal{L}}^{-\top}}{2} \biggr) U^{\top} \implies
        (\mathcal{L}^{\dagger}_s)^{\dagger}
        =
        U \biggl( \frac{\bar{\mathcal{L}}^{-1} + \bar{\mathcal{L}}^{-\top}}{2} \biggr)^{-1} U^{\top}
        =
        U \bar{\mathcal{L}} (\bar{\mathcal{L}}_s)^{-1} \bar{\mathcal{L}}^{\top} U^{\top}.
    \end{equation*}
    On the other hand, the same lemma gives
    \begin{equation*}
        \mathcal{L} (\mathcal{L}_s)^{\dagger} \mathcal{L}^{\top}
        =
        (U \bar{\mathcal{L}} U^{\top}) (U (\bar{\mathcal{L}}_s)^{-1} U^{\top}) (U \bar{\mathcal{L}}^{\top} U^{\top})\\
        =
        U \bar{\mathcal{L}} (\bar{\mathcal{L}}_s)^{-1} \bar{\mathcal{L}}^{\top}
        U^{\top}.
    \end{equation*}
    Comparing the two expressions yields $(\mathcal{L}^{\dagger}_s)^{\dagger} = \mathcal{L} (\mathcal{L}_s)^{\dagger} \mathcal{L}^{\top}$. 
    
    Next, recall that $\mathcal{L} = \mathcal{L}_s + \mathcal{K}$ and $\mathcal{L}^{\top} = \mathcal{L}_s - \mathcal{K}$. We obtain
    \begin{align*}
        \mathcal{L} (\mathcal{L}_s)^{\dagger} \mathcal{L}^{\top}
        =
        (\mathcal{L}_s+\mathcal{K}) (\mathcal{L}_s)^{\dagger} (\mathcal{L}_s-\mathcal{K})
        =
        \mathcal{L}_s - \mathcal{K} (\mathcal{L}_s)^{\dagger} \mathcal{K}
        =
        \mathcal{L}_s + \mathcal{K} (\mathcal{L}_s)^{\dagger} \mathcal{K}^{\top}.
    \end{align*}
    
    For every $\bm{x} \in \mathbb{R}^n$, we have $\bm{x}^{\top} \mathcal{K}(\mathcal{L}_s)^{\dagger}\mathcal{K}^{\top}\bm{x} = (\mathcal{K}^{\top}\bm{x})^{\top} (\mathcal{L}_s)^{\dagger} (\mathcal{K}^{\top}\bm{x}) \geq 0$. Hence, $\mathcal{K} (\mathcal{L}_s)^{\dagger} \mathcal{K}^{\top} \succeq O$, which proves $(\mathcal{L}^{\dagger}_s)^{\dagger} \succeq\mathcal{L}_s$.

    We finally characterize the equality condition. Suppose that $\mathcal{K}(\mathcal{L}_s)^{\dagger}\mathcal{K}^{\top} = O$.  $(\mathcal{K}^{\top}\bm{x})^{\top} (\mathcal{L}_s)^{\dagger} (\mathcal{K}^{\top}\bm{x}) = 0 \ (\forall \bm{x} \in \mathbb{R}^n)$ implies that 
    $\mathcal{K}^{\top}\bm{x} \in \ker((\mathcal{L}_s)^{\dagger})$, and therefore $\operatorname{Im}(\mathcal{K}^{\top}) \subseteq \ker((\mathcal{L}_s)^{\dagger}) =\operatorname{span}(\bm{1})$. Meanwhile, $\mathcal{K}\bm{1}=\bm{0}$ implies $\operatorname{Im}(\mathcal{K}^{\top}) \subseteq \operatorname{span}(\bm{1})^{\perp}$. Therefore, $\mathcal{K} = O$, and consequently $\mathcal{L}$ is symmetric. The converse is immediate. 
\end{proof}

The Loewner inequality above immediately yields a comparison of the effective resistances associated with the two unsigned SCWB-directed graphs.

\begin{corollary}
    \label{cor:comparison_of_effective_resistances}
    Let $\mathcal{L}$ be an unsigned SCWB-directed Laplacian. Then, 
    \begin{equation}
        \label{eq:comparison_of_effective_resistances}
        R_{\mathcal{G}(\mathcal{L})}(i,j) 
        = R_{\mathcal{G}((\mathcal{L}^{\dagger}_s)^{\dagger})}(i,j) 
        \leq 
        R_{\mathcal{G}(\mathcal{L}_s)}(i,j)
        \quad \forall i, j \in [n],\ i \neq j.
    \end{equation}
    Moreover, the equality holds for every pair $i,j\in[n]$ if and only if $\mathcal{L}$ is symmetric.
\end{corollary}
\begin{proof}
    The equality on the left simply reflects the fact that the two graphs $\mathcal{G}(\mathcal{L})$ and $\mathcal{G}((\mathcal{L}^{\dagger}_s)^{\dagger})$ have the same effective resistance.

    We now demonstrate the inequality on the right. Given a matrix $U \in \mathbb{R}^{n \times (n-1)}$ whose orthonormal columns span $\bm{1}^{\perp}$, we have $U^{\top} (\mathcal{L}^{\dagger}_s)^{\dagger}U \succeq U^{\top} \mathcal{L}_s U$. Since inversion reverses the Loewner order for positive-definite matrices, $(U^{\top} (\mathcal{L}^{\dagger}_s)^{\dagger}U)^{-1} \preceq (U^{\top} \mathcal{L}_s U)^{-1}$. Therefore, 
    \begin{equation*}
        \mathcal{L}^{\dagger}_s = U (U^{\top} (\mathcal{L}^{\dagger}_s)^{\dagger}U)^{-1} U^{\top} \preceq U (U^{\top} \mathcal{L}_s U)^{-1} U^{\top} = (\mathcal{L}_s)^{\dagger}.
    \end{equation*}
    Evaluating this matrix inequality at $\bm{1}_i - \bm{1}_j$ gives 
    \begin{align*}
        R_{\mathcal{G}(\mathcal{L}_s)}(i,j) - R_{\mathcal{G}((\mathcal{L}^{\dagger}_s)^{\dagger})}(i,j) 
        = (\bm{1}_i - \bm{1}_j)^{\top} ( (\mathcal{L}_s)^{\dagger} -\mathcal{L}^{\dagger}_s ) (\bm{1}_i - \bm{1}_j)
        \geq 0.
    \end{align*}

    Suppose that equality holds for every pair of $i,j \in [n]$. Let $\mathcal{L}_d := (\mathcal{L}_s)^{\dagger} - \mathcal{L}^{\dagger}_s$. Then, $(\bm{1}_i-\bm{1}_j)^{\top} \mathcal{L}_d (\bm{1}_i-\bm{1}_j) = 0 \ (\forall i, j \in [n])$. Since $\mathcal{L}_d \succeq O$, this implies $\mathcal{L}_d (\bm{1}_i - \bm{1}_j) = \bm{0}$ for every $i \neq j$. The vectors $\{\bm{1}_i-\bm{1}_j\}$ span $\operatorname{span}(\bm{1})^{\perp}$, while $\mathcal{L}_d\bm{1} = \bm{0}$. Hence, $\mathcal{L}_d = O$, that is, $(\mathcal{L}_s)^{\dagger} = \mathcal{L}^{\dagger}_s$. Taking pseudoinverses and then applying the equality condition in \cref{thm:Ldaggersdagger_decomposition} shows that $\mathcal{L}$ is symmetric.
    The converse is immediate.
\end{proof}

\begin{remark}[Relation to the existing result of \cite{Fontan2023pseudoinverse}]
    \label{rem:comparison_with_Fontan_Altafini}
    For a normal signed Laplacian $\mathcal{L}$ such that $-\mathcal{L}$ is eventually exponentially positive (EEP)\footnote{
        Since $-\mathcal{L}$ is an irreducible Metzler matrix for an unsigned SC-directed graph, $e^{-\mathcal{L}t} > O$ holds for every $t > 0$. Hence, the EEP assumption is automatically satisfied in the present unsigned setting.
    }, Fontan and Altafini established the corresponding effective-resistance inequality; see Lem.V.4 of \cite{Fontan2023pseudoinverse}.
    
    In the unsigned SC-directed setting, normality implies that the associated graph is weight-balanced. Indeed, normality gives $\|\mathcal{L}^{\top}\bm{1}\|_2^2 = \bm{1}^{\top} \mathcal{L} \mathcal{L}^{\top} \bm{1} =
    \bm{1}^{\top} \mathcal{L}^{\top} \mathcal{L} \bm{1} = \bm{0}$. Hence, $\mathcal{L}^{\top}\bm{1}=\bm{0}$. Therefore, unsigned normal SC-directed Laplacians form a subclass of unsigned SCWB-directed Laplacians. Moreover, this inclusion is strict. Consider the unsigned SCWB-directed Laplacian in \cref{eq:exmaple_SCWB_directed_to_signed_undirected}. A direct calculation gives
    \begin{equation*}
        \mathcal{L}\mathcal{L}^{\top}
        -
        \mathcal{L}^{\top}\mathcal{L}
        =
        \begin{bmatrix}
            -4 &  2 &  2 &  0 \\
             2 &  0 &  0 & -2 \\
             2 &  0 &  0 & -2 \\
             0 & -2 & -2 &  4
        \end{bmatrix}
        \neq O.
    \end{equation*}
    Hence, $\mathcal{L}$ is not normal. This proves that the class of unsigned normal SC-directed Laplacians is a proper subclass of the class of unsigned SCWB-directed Laplacians.

    Consequently, \cref{cor:comparison_of_effective_resistances} extends the normality-based effective-resistance comparison of Lem.V.4 of \cite{Fontan2023pseudoinverse} to all unsigned SCWB-directed Laplacians.
\end{remark}

The decomposition in \cref{thm:Ldaggersdagger_decomposition} not only extends the effective-resistance comparison beyond the normal case,
but also provides an explicit matrix representation of the contribution of graph directionality: $(\mathcal{L}^{\dagger}_s)^{\dagger}-\mathcal{L}_s = \mathcal{K}(\mathcal{L}_s)^{\dagger}\mathcal{K}^{\top} \succeq O$.


\subsection{Commutative diagram}
For SCWB-directed graphs, effective resistance is invariant under Kron reduction (Lem.I\hspace{-1.2pt}V.4 of \cite{Sugiyama2023graph}).
Together with the generalized FBI, this yields a commutative relation between two operations on the Laplacian.

\begin{theorem}
  \label{thm:com_Ldagger_s_dagger_KR}
  Let $\mathcal{L}$ 
  be the Laplacian corresponding to an SCWB-directed graph. 
  For any set of nodes $\alpha$ such that $\alpha \subsetneq [n],\ |\alpha| \geq 2$, the following equality holds:
  \begin{equation}\label{eq:com_Ldagger_s_dagger_KR}
    (\mathcal{L}^{\dagger}_s)^{\dagger} / \alpha^c
    =
    ((\mathcal{L} /\alpha^c)^{\dagger}_s)^{\dagger}.
  \end{equation}
  In other words, taking the pseudoinverse of the symmetrized pseudoinverse, $\mathcal{L} \mapsto (\mathcal{L}^{\dagger}_s)^{\dagger}$, commutes with Kron reduction, $\mathcal{L} \mapsto \mathcal{L} / \alpha^c$.
\end{theorem}

\begin{proof}
    Taking inverses in the FBI \cref{eq:FBI_generalized}, we obtain
    \begin{equation}\label{eq:FBI_inverse_1}
        -\frac{1}{2}
        \begin{bmatrix}
            0 &\bm{1}^{\top} \\
            \bm{1} & \Omega
        \end{bmatrix}
        =
        \begin{bmatrix}
            4\sigma^2 & -2\bm{p}^{\top} \\ 
            -2\bm{p} & (\mathcal{L}^{\dagger}_s)^{\dagger}
        \end{bmatrix}^{-1}
        \text{ where } 
        \begin{cases}
            \bm{p} = \dfrac{1}{2} (\mathcal{L}^{\dagger}_s)^{\dagger} \bm{\zeta} + \dfrac{1}{n} \bm{1}, \\
            \sigma^2 = \dfrac{1}{4} \bm{\zeta}^{\top} (\mathcal{L}^{\dagger}_s)^{\dagger} \bm{\zeta} + \dfrac{1}{n} \bm{1}^{\top} \bm{\zeta}.
        \end{cases}
    \end{equation}
    Here, the augmented matrix is indexed by $\{0\}\cup[n]$, where the index $0$ corresponds to the first row and column. 
    We calculate the submatrix corresponding to $\alpha^+ := \{0\} \cup \alpha$ on the right-hand side of \cref{eq:FBI_inverse_1}:
    \begin{equation}\label{eq:FBI_inverse_2}
        \begin{aligned}
            &
            \begin{bmatrix}
              4\sigma^2 & -2\bm{p}^{\top} \\
              -2\bm{p} & (\mathcal{L}^{\dagger}_s)^{\dagger}
            \end{bmatrix}^{-1}
            [\alpha^+, \alpha^+]
            \underset{\text{\cref{eq:FBI_inverse_1}}}{=} 
            -\frac{1}{2}
            \begin{bmatrix}
              0 & \bm{1}^{\top} \\
              \bm{1} & \Omega
            \end{bmatrix}
            [\alpha^+, \alpha^+] 
            =
            -\dfrac{1}{2}
            \begin{bmatrix}
              0 & \bm{1}[\alpha]^{\top} \\
              \bm{1}[\alpha] & \Omega[\alpha, \alpha]
            \end{bmatrix} 
            \\
            &
            =
            -\dfrac{1}{2}
            \begin{bmatrix}
              0 & \bm{1}[\alpha]^{\top} \\
              \bm{1}[\alpha] & \Omega_{\text{red}}
            \end{bmatrix} 
            =
            \begin{bmatrix}
                4\sigma^2_{\text{red}} & -2\bm{p}_{\text{red}}^{\top} \\
                -2\bm{p}_{\text{red}} & ((\mathcal{L} / \alpha^c )^{\dagger}_s)^{\dagger}
            \end{bmatrix}^{-1}
            .
        \end{aligned}
    \end{equation}
    The third equality holds since Kron reduction preserves the effective resistance (Lem.I\hspace{-1.2pt}V.4 of \cite{Sugiyama2023graph}). 
    The last equation is the FBI itself in the Kron-reduced graph $\mathcal{G}_{\text{red}} ( =\mathcal{G}(\mathcal{L}/ \alpha^c))$. 
    Here, $\bm{p}_{\text{red}} \in \mathbb{R}^{|\alpha|}$ and ${\sigma}_{\text{red}}^2 \in \mathbb{R}$ are the resistance curvature and resistance radius in $\mathcal{G}_{\text{red}}$, respectively. 

    Next, we calculate the inverse matrix of the leftmost side of \cref{eq:FBI_inverse_2}:
    \begin{align*}
    &
    \Biggl[
        \begin{bmatrix}
            4\sigma^2 & -2\bm{p}^{\top} \\
            -2\bm{p} & (\mathcal{L}^{\dagger}_s)^{\dagger}
        \end{bmatrix}^{-1}
        [\alpha^+, \alpha^+] 
    \Biggr]^{-1} 
    =
    \begin{bmatrix}
      4\sigma^2 & -2\bm{p}^{\top} \\
      -2\bm{p} & (\mathcal{L}^{\dagger}_s)^{\dagger}
    \end{bmatrix}
    \biggl/ 
    \begin{bmatrix}
      4\sigma^2 & -2\bm{p}^{\top} \\
      -2\bm{p} & (\mathcal{L}^{\dagger}_s)^{\dagger}
    \end{bmatrix} [\alpha^c, \alpha^c]
    \\
    &= 
    \begin{bmatrix}
      4\sigma^2 & -2\bm{p}[\alpha]^{\top} \\
      -2\bm{p}[\alpha] & (\mathcal{L}^{\dagger}_s)^{\dagger}[\alpha, \alpha]
    \end{bmatrix}
    -
    \begin{bmatrix}
      -2\bm{p}[\alpha^c]^{\top} \\
      (\mathcal{L}^{\dagger}_s)^{\dagger}[\alpha, \alpha^c]
    \end{bmatrix}
    (\mathcal{L}^{\dagger}_s)^{\dagger}[\alpha^c, \alpha^c]^{-1}
    \begin{bmatrix}
      -2\bm{p}[\alpha^c] &
      (\mathcal{L}^{\dagger}_s)^{\dagger}[\alpha^c, \alpha]
    \end{bmatrix}
    .
    \end{align*} 
    We use the following Schur-complement identity: $A^{-1}[\beta] = (A / \beta^c)^{-1} \  (\forall\beta \subset [n])$ (Thm.I.2 of \cite{Zhang2005SchurComplement}) in the first equality. The last expression above corresponds to the right-hand side of the FBI in $\mathcal{G}_{\text{red}}$. Focusing on the bottom right block, we have
    \begin{equation*}\label{eq:commutative_diagram}
        ((\mathcal{L} / \alpha^c)^{\dagger}_s)^{\dagger}
        = 
        (\mathcal{L}^{\dagger}_s)^{\dagger}[\alpha, \alpha] 
        - 
        (\mathcal{L}^{\dagger}_s)^{\dagger}[\alpha, \alpha^c] 
        (\mathcal{L}^{\dagger}_s)^{\dagger}[\alpha^c, \alpha^c]^{-1}
        (\mathcal{L}^{\dagger}_s)^{\dagger}[\alpha^c, \alpha] \\
        =
        (\mathcal{L}^{\dagger}_s)^{\dagger} / \alpha^c.
    \end{equation*}
    This shows that the two algebraic operations on $\mathcal{L}$, namely
    taking the pseudoinverse of the symmetrized pseudoinverse, $\mathcal{L} \mapsto (\mathcal{L}^{\dagger}_s)^{\dagger}$, and Kron reduction, $\mathcal{L} \mapsto \mathcal{L} / \alpha^c$, are commutative, as required.
\end{proof}

This commutativity is not apparent from the FBI for undirected graphs (Thm.I\hspace{-1.2pt}I\hspace{-1.2pt}I.6 of \cite{Devriendt2022graph}) due to the fact that $(\mathcal{L}^{\dagger}_s)^{\dagger} - \mathcal{L} = (\mathcal{L}^{\dagger}_s)^{\dagger} - \mathcal{L}_s = O$ in that case. 
Furthermore, we note that taking the symmetrized pseudoinverse, $\mathcal{L} \mapsto \mathcal{L}^{\dagger}_s$, and Kron reduction, $\mathcal{L} \mapsto \mathcal{L} / \alpha^c$, do not commute. In other words, 
\begin{equation}
    \label{eq:Not_commutaitve_diagram}
    \mathcal{L}^{\dagger}_s / \alpha^c \neq (\mathcal{L} / \alpha^c )^{\dagger}_s \quad \text{in general}.
\end{equation}

\begin{example}
    For the Laplacian in \cref{eq:exmaple_SCWB_directed_to_signed_undirected}, 
    let $\alpha = \{1, 3, 4\}$ be the set of nodes to be retained by Kron reduction. Then the two matrices are as follows:
    \begin{align*}
        \mathcal{L}^{\dagger}_s / \alpha^c
        =
        \frac{1}{330}
        \begin{bmatrix}
            41 & -27 & -14 \\
            -27 & 54 & -27 \\
            -14 & -27 & 41
        \end{bmatrix}
        \neq
        (\mathcal{L} / \alpha^c )^{\dagger}_s
        = \frac{1}{54}
        \begin{bmatrix}
            7 & -5 & -2 \\
            -5 & 10 & -5 \\
            -2 & -5 & 7
        \end{bmatrix}
    \end{align*}
\end{example}

Finally, we show that weight-balancing with a positive diagonal matrix and Kron reduction commute for SC-directed graphs.

\begin{lemma}
\label{lem:commutative_diagram_Kron_redution_and_weight_balanced}
    Let $\mathcal{L}$ be the Laplacian of an SC-directed graph, and let $M$ be a positive diagonal matrix such that $M\mathcal{L}$ is weight-balanced. For any set of nodes $\alpha$ such that $\alpha \subsetneq [n],\ |\alpha|\geq 2$, the following equality holds:
    \begin{equation}
        \label{eq:com_weight_balanced_KR}
    (M \mathcal{L}) / \alpha^c = M[\alpha](\mathcal{L} /\alpha^c).
  \end{equation}
  In other words, making the Laplacian weight-balanced, $\mathcal{L} \mapsto M\mathcal{L}$, commutes with Kron reduction, $\mathcal{L} \mapsto \mathcal{L}/ \alpha^c$.
\end{lemma}
\begin{proof}
    Since the submatrix $\mathcal{L}[\alpha^c, \alpha^c]$ is invertible (see the proof of Lem.I\hspace{-1.2pt}I\hspace{-1.2pt}I.3 of \cite{Sugiyama2023graph}), it can be shown by direct calculation that the two sides are equal.
    \begin{align*}
        (M \mathcal{L}) / \alpha^c 
        &=
        \begin{bmatrix}
            M[\alpha] \mathcal{L}[\alpha, \alpha] &  M[\alpha] \mathcal{L}[\alpha, \alpha^c]\\
            M[\alpha^c] \mathcal{L}[\alpha^c, \alpha] & M[\alpha^c] \mathcal{L}[\alpha^c, \alpha^c]
        \end{bmatrix}
        \biggl/  \alpha^c \\
        &= M[\alpha] \mathcal{L}[\alpha, \alpha] - M[\alpha] \mathcal{L}[\alpha, \alpha^c] (M[\alpha^c] \mathcal{L}[\alpha^c, \alpha^c])^{-1} M[\alpha^c] \mathcal{L}[\alpha^c, \alpha] \\
        &= M[\alpha] (\mathcal{L}[\alpha, \alpha] - \mathcal{L}[\alpha, \alpha^c]  \mathcal{L}[\alpha^c, \alpha^c]^{-1} \mathcal{L}[\alpha^c, \alpha]) \\
        &= M[\alpha] (\mathcal{L}/ \alpha^c).
    \end{align*}
    Since $(M\mathcal{L})/\alpha^c$ is an SCWB-directed Laplacian, the above equality shows that the matrix $M[\alpha]$ weight-balances $\mathcal{L}/\alpha^c$.
\end{proof}

This lemma shows that the order of weight-balancing and Kron reduction does not affect the resulting reduced SCWB-directed Laplacian. This commutativity is useful because properties and formulas established for weight-balanced directed graphs can be transferred to general SC-directed graphs before or after Kron reduction (discussed in more detail in the next subsection).

\subsection{Resistance curvature and resistance radius}
\label{subsec:resistance_curvature_resistance_radius}
We now extend the concepts of resistance curvature and resistance radius to the class of SC-directed graphs. For any SC-directed graph $\mathcal{G}(\mathcal{L})$, there exists a corresponding weight-balanced Laplacian $\mathcal{L}_{\text{WB}}$,
as mentioned in \Cref{subsec:ER}.
The FBI for SC-directed graphs \cref{eq:FBI_simple_SC} establishes a formal relationship between their Laplacians and resistance matrices, providing a consistent foundation for defining these geometric invariants in the SC-directed setting.

\begin{definition}
    \label{def:resistance_curvature_radius_SC}
    Let $\mathcal{G}(\mathcal{L})$ be an SC-directed graph and let $\Omega$ be the resistance matrix.
    Let $M$ be a positive diagonal matrix such that $M\mathcal{L}$ is an SCWB-directed Laplacian, and define $\Omega_{\mathrm{WB}} := M^{-1} \Omega$. 
    Suppose that $\bm{1}^{\top}\Omega^{-1}\bm{1} > 0$. 
    Then, the resistance curvature $\bm{p}$ and resistance radius $\sigma^2$ of $\mathcal{G}(\mathcal{L})$ are defined by
      \begin{equation*}
        \bm{p} 
        := \frac{\Omega^{-1} \bm{1}}{\bm{1}^{\top} \Omega^{-1} \bm{1}} 
        = \frac{\Omega_{\text{WB}}^{-1} M^{-1} \bm{1}}{\bm{1}^{\top} \Omega_{\text{WB}}^{-1} M^{-1} \bm{1}}, \quad
        \sigma^2 
        := \frac{1}{2} (\bm{1}^{\top} \Omega^{-1} \bm{1})^{-1}
        = \frac{1}{2}(\bm{1}^{\top} \Omega_{\text{WB}}^{-1} M^{-1} \bm{1})^{-1}.
    \end{equation*}
\end{definition}
\begin{remark}
    The positivity condition is automatically satisfied when the graph is SCWB-directed. Indeed, in that case, $\bm{1}^{\top}\Omega^{-1}\bm{1}=1/(2\sigma^2)>0$ follows from the generalized Fiedler--Bapat identity.
    
    For a general SC-directed graph, however, this positivity is not automatic. Let $\bm{p}_{\mathrm{WB}}$ and $\sigma_{\mathrm{WB}}^2$ denote the resistance curvature and resistance radius of $\mathcal{G}(\mathcal{L}_{\mathrm{WB}})$.
    Since $\Omega = M\Omega_{\mathrm{WB}}$, we have $\bm{1}^{\top}\Omega^{-1}\bm{1} = \bm{p}_{\mathrm{WB}}^{\top}M^{-1}\bm{1} / 2\sigma_{\mathrm{WB}}^2$. Therefore, the condition in Definition~\ref{def:resistance_curvature_radius_SC}
    is equivalent to $\bm{p}_{\mathrm{WB}}^{\top}M^{-1}\bm{1}>0$. In particular, the positive diagonal scaling $M$ may cause the positive and
    negative components of $\bm{p}_{\mathrm{WB}}$ to cancel, so the quantity $\bm{1}^{\top}\Omega^{-1}\bm{1}$ may be zero or negative.
\end{remark}

The resulting quantities are graph invariants that encode the effect of the scaling matrix $M$. Note that $\Omega$ is invertible since $\Omega_{\text{WB}}$ is invertible, as proven in \cref{thm:FBI_generalized}, and thus $\det \Omega = \det M \cdot \det \Omega_{\text{WB}} \neq 0$. Since $\Omega$ is uniquely determined by $\mathcal{G}(\mathcal{L})$, the quantities $\bm{p}$ and $\sigma^2$ are uniquely determined whenever $\bm{1}^{\top}\Omega^{-1}\bm{1}>0$. The equivalent expressions involving $M$ and $\Omega_{\mathrm{WB}}$ are independent of the choice of the positive scalar normalization of $M$. Furthermore, the relationship between resistance curvature and resistance radius in $\mathcal{G}(\mathcal{L})$ and $\mathcal{G}(\mathcal{L}_{\text{WB}})$ is as follows.

\begin{lemma}
    \label{lem:resistance_curvature/radius_WB}
    Let $\mathcal{G}(\mathcal{L})$ be an SC-directed graph with resistance matrix $\Omega$, and let $M$ be a positive diagonal matrix such that
    $\mathcal{L}_{\mathrm{WB}}:=M\mathcal{L}$ is an SCWB-directed Laplacian.
    Assume that $\bm{1}^{\top}\Omega^{-1}\bm{1}>0$. Let $\bm{p}$ and $\sigma^2$ be the resistance curvature and resistance radius
    of $\mathcal{G}(\mathcal{L})$, respectively, as defined in \cref{def:resistance_curvature_radius_SC}.
    Furthermore, let $\bm{p}_{\mathrm{WB}}$ and $\sigma_{\mathrm{WB}}^2$ be those of $\mathcal{G}(\mathcal{L}_{\mathrm{WB}})$. Then, 
    \begin{equation}
        \label{eq:resistance_curvature/radius_WB}
        \begin{cases}
            \displaystyle \bm{p}_{\text{WB}} = \frac{\Omega^{-1} M \Omega \bm{p}}{\bm{1}^{\top} \Omega^{-1} M \Omega \bm{p}} \\
            \displaystyle \sigma^2_{\text{WB}} = \frac{\sigma^2}{\bm{1}^{\top} \Omega^{-1} M \Omega \bm{p}}
        \end{cases}, \quad
        \begin{cases}
            \displaystyle \bm{p} = \frac{\Omega^{-1} M^{-1} \Omega \bm{p}_{\text{WB}}}{\bm{1}^{\top} \Omega^{-1} M^{-1} \Omega \bm{p}_{\text{WB}}} \\
            \displaystyle \sigma^2 = \frac{\sigma^2_{\text{WB}}}{\bm{1}^{\top} \Omega^{-1} M^{-1} \Omega \bm{p}_{\text{WB}}}
        \end{cases}.
    \end{equation}
\end{lemma}
\begin{proof}
    Using the relationship $\Omega \bm{p}= 2 \sigma^2 \bm{1}$ from \cref{def:resistance_curvature_radius_SC} 
    and $\Omega_{\text{WB}} \bm{p}_{\text{WB}} = 2\sigma^2_{\text{WB}} \bm{1}$ from \cref{cor:FBI_generalized}, we obtain
    \begin{align*}
        &
        \bm{p}_{\text{WB}}
        = \frac{\Omega_{\text{WB}}^{-1} \bm{1}}{\bm{1}^{\top} \Omega_{\text{WB}}^{-1} \bm{1}}
        = \frac{\Omega^{-1} M \ ( \frac{1}{2\sigma^2} \Omega \bm{p})}{\bm{1}^{\top}\Omega^{-1} M \ (\frac{1}{2\sigma^2} \Omega \bm{p})}
        = \frac{\Omega^{-1} M \Omega \bm{p}}{\bm{1}^{\top} \Omega^{-1} M \Omega \bm{p}}, \\
        &
        \sigma^2_{\text{WB}}
        = \frac{1}{2 (\bm{1}^{\top} \Omega_{\text{WB}}^{-1} \bm{1})}
        = \frac{1}{2 (\bm{1}^{\top} \Omega^{-1} M \ (\frac{1}{2\sigma^2} \Omega \bm{p}))} 
        = \frac{\sigma^2}{\bm{1}^{\top} \Omega^{-1} M \Omega \bm{p}}.
    \end{align*}
    The proof is similar for the inverse relationship.
\end{proof}

Through the relationships above, we can investigate the changes in resistance curvature and resistance radius between $\mathcal{G}(\mathcal{L})$ and $\mathcal{G}(\mathcal{L}_{\text{WB}})$. 
Furthermore, these changes before and after Kron reduction in an SCWB-directed graph can also be obtained from the proof of \cref{thm:com_Ldagger_s_dagger_KR} as follows:
\begin{equation} 
    \label{eq:change_resistance_curvature/radius}
    \begin{gathered}
      \bm{p}_{\text{WB.red}} = \bm{p}_{\text{WB}}[\alpha] - (\mathcal{L}^{\dagger}_s)^{\dagger}[\alpha, \alpha^c] (\mathcal{L}^{\dagger}_s)^{\dagger}[\alpha^c, \alpha^c]^{-1} \bm{p}_{\text{WB}}[\alpha^c], \\
      \sigma_{\text{WB.red}}^2 = \sigma^2_{\text{WB}} - \bm{p}_{\text{WB}}[\alpha^c]^{\top} (\mathcal{L}^{\dagger}_s)^{\dagger}[\alpha^c, \alpha^c]^{-1} \bm{p}_{\text{WB}}[\alpha^c].
    \end{gathered}
\end{equation}
By combining \cref{lem:resistance_curvature/radius_WB} and \cref{eq:change_resistance_curvature/radius}, we can track the change in resistance curvature and resistance radius before and after Kron reduction even in SC-directed graphs.

We first examine how resistance curvature and resistance radius change when the graph is SCWB-directed. According to equation \cref{eq:FBI_generalized_relations_1}, the sum of the node-wise resistance curvatures is 1. This property is inherited after Kron reduction:
\begin{align*}
  \bm{1}[\alpha]^{\top} \bm{p}_{\text{red}}
  &= \bm{1}[\alpha]^{\top} ( \bm{p}[\alpha] - (\mathcal{L}^{\dagger}_s)^{\dagger}[\alpha, \alpha^c] (\mathcal{L}^{\dagger}_s)^{\dagger}[\alpha^c, \alpha^c]^{-1} \bm{p}[\alpha^c] ) \\
  &= \bm{1}[\alpha]^{\top} \bm{p}[\alpha] + \bm{1}[\alpha^c]^{\top} (\mathcal{L}^{\dagger}_s)^{\dagger}[\alpha^c, \alpha^c] (\mathcal{L}^{\dagger}_s)^{\dagger}[\alpha^c, \alpha^c]^{-1} \bm{p}[\alpha^c] \quad (\because \bm{1}^{\top} {(\mathcal{L}^{\dagger}_s)^{\dagger}} = \bm{0}^{\top}) \\
  &= \bm{1}[\alpha]^{\top} \bm{p}[\alpha] + \bm{1}[\alpha^c]^{\top} \bm{p}[\alpha^c] \\
  &= 1.
\end{align*}
As a simple case, consider $\alpha = [n-1]$. From \cref{eq:change_resistance_curvature/radius}, we have 
\begin{equation}\label{eq:change_p_sigma2_KR_one_node}
  p_{\text{red}.i} = p_i - \frac{(\mathcal{L}^{\dagger}_s)^{\dagger}_{in}}{(\mathcal{L}^{\dagger}_s)^{\dagger}_{nn}} p_n, \quad
  \sigma_{\text{red}}^2 = \sigma^2 - \frac{p_n^2}{(\mathcal{L}^{\dagger}_s)^{\dagger}_{nn}} \leq \sigma^2 
  \quad  \text{for all } i \in [n-1].
\end{equation}
We note that $(\mathcal{L}^{\dagger}_s)^{\dagger}_{nn} = \bm{1}_n^{\top} (\mathcal{L}^{\dagger}_s)^{\dagger} \bm{1}_n > 0$ holds since $(\mathcal{L}^{\dagger}_s)^{\dagger}$ is PSD with its kernel being $\operatorname{span}(\bm{1})$.
When the resistance curvature at node $n$ is nonzero, the resistance radius strictly decreases in the Kron-reduced graph $\mathcal{G}(\mathcal{L}/ \{n\} )$. 
Therefore, the resistance radius is non-decreasing when viewed as a function on the node set, as in the case of unsigned undirected graphs (Def.I\hspace{-1.2pt}I\hspace{-1.2pt}I.37 of \cite{Devriendt2022graph}). 

\begin{definition}[resistance radius]
  \label{def:resistance_radius_for_SCWB}
  The resistance radius (set) function for an SCWB-directed graph is defined by 
  \begin{equation}
          \sigma^2(\mathcal{V}) 
          := \dfrac{1}{2}(\bm{1}^{\top}\Omega[\mathcal{V}, \mathcal{V}]^{-1} \bm{1})^{-1} 
          \text{ for all } \mathcal{V} \subseteq [n] \text{ and } |\mathcal{V}| \geq 2
  \end{equation}
 and we define $\sigma^2 (\emptyset) = \sigma^2 (\{i \}) = 0$ for all $i \in [n]$.
\end{definition}

Note that the effective resistance in the SCWB-directed case is given by \cref{eq:ER_SCWB}. Under this definition, the following properties hold:

\begin{theorem}
    \label{thm:property_resistance_radius_set_function}
    The following properties hold for the resistance radius set function of an SCWB-directed graph:
    \begin{enumerate} 
    \item 
    For any node sets $\mathcal{U}, \mathcal{V} \subseteq [n]$ satisfying 
    $\mathcal{U} \subsetneq \mathcal{V}$ and $| \mathcal{U}| \geq 2$, 
    \begin{equation*}
        0 < \sigma^2(\mathcal{U}) \leq \sigma^2(\mathcal{V}).
    \end{equation*}
    Moreover, the equality holds if and only if $\bm{p}(\mathcal{V})[\mathcal{V}\backslash\mathcal{U}] = \bm{0}$
    where $\bm{p}(\mathcal{V})$ is the resistance curvature of the Kron-reduced graph $\mathcal{G}(\mathcal{L}/ \mathcal{V}^c)$. 
    \item 
    For any pair of nodes $i, j \in [n]$ satisfying $i \neq j$, 
    \begin{equation*}
        \sigma^2(\{i,j\}) = R_{\mathcal{G}}(i,j)/4 > 0.
    \end{equation*}
    \end{enumerate}
\end{theorem}
\begin{proof}
    Let $\beta := \mathcal{V}\backslash \mathcal{U}$. Consider the Kron-reduced graph $\mathcal{G}(\mathcal{L} / \mathcal{V}^c)$. Reducing this graph further onto $\mathcal{U}$ amounts to eliminating the node set $\beta$. From \cref{eq:change_resistance_curvature/radius}, we have 
    \begin{equation}
        \label{eq:resistance_radius_non_decresing_function}
        \sigma^2(\mathcal{U}) = \sigma^2(\mathcal{V}) - \bm{p}(\mathcal{V})[\beta]^\top ((\mathcal{L}/\mathcal{V}^c)^{\dagger}_s)^\dagger [\beta,\beta]^{-1} \bm{p}(\mathcal{V})[\beta]. 
    \end{equation}
    Here, the matrix $((\mathcal{L}/\mathcal{V}^c)^{\dagger}_s)^{\dagger}$ is PSD with its kernel being $\operatorname{span}(\bm{1})$. Since $\beta$ is a proper subset of $\mathcal{V}$, the principal submatrix $((\mathcal{L}/\mathcal{V}^c)^{\dagger}_s)^\dagger [\beta,\beta]$ is positive-definite. Indeed, for any $\bm{x} \in \mathbb{R}^{|\beta|} \backslash \{\bm{0} \}$, let $\tilde{\bm{x}} \in \mathbb{R}^{|\mathcal{V}|}$ be the vector obtained by extending $\bm{x}$ by zeros on $\mathcal{U}$. Then $\tilde{\bm{x}} \notin \operatorname{span}(\bm{1})$, and hence 
    \begin{equation*}
        \bm{x}^\top ((\mathcal{L}/\mathcal{V}^c)^{\dagger}_s)^{\dagger} [\beta, \beta] \bm{x} = \tilde{\bm{x}}^\top ((\mathcal{L}/\mathcal{V}^c)^{\dagger}_s)^{\dagger} \tilde{\bm{x}} > 0. 
    \end{equation*}
    Therefore, the quadratic term on the right-hand side of \cref{eq:resistance_radius_non_decresing_function} is non-negative, which implies 
    $\sigma^2(\mathcal{U}) \leq \sigma^2(\mathcal{V})$. 
    Moreover, since $((\mathcal{L}/\mathcal{V}^c)^{\dagger}_s)^{\dagger} [\beta,\beta]$ is positive-definite, the equality holds if and only if $\bm{p}(\mathcal{V})[\beta]=\bm{0}$.
    
    Regarding the second point, the Laplacian of a Kron-reduced graph with only two remaining nodes $i, j \in [n] \ (i \neq j)$ can be written as $c \bigl[\begin{array}{cc}\begin{smallmatrix}1 & -1 \\ -1 & 1 \end{smallmatrix}\end{array}\bigr]$, 
    where $c$ is some positive constant. The pseudoinverse can be calculated directly, and we find that 
    $\mathcal{L}^{\dagger} = \frac{1}{4c} \bigl[\begin{array}{cc}\begin{smallmatrix}1 & -1 \\ -1 & 1 \end{smallmatrix}\end{array}\bigr]$. 
    Therefore, the effective resistance between $i$ and $j$ is $R_{\mathcal{G}}(i, j) = 1/c$, and we obtain that $\sigma^2(\{i,j \}) = 1/4c = R_{\mathcal{G}}(i, j) /4 > 0$. This completes the proof.
\end{proof}

This monotonicity and the positivity of the resistance radius are consistent with the geometric interpretation of the resistance radius as a size-related global graph invariant. However, this monotonicity is specific to the SCWB-directed setting and does not extend to general SC-directed cases.

\begin{figure}[t]
  \centering
  \includegraphics[width=0.65\textwidth]{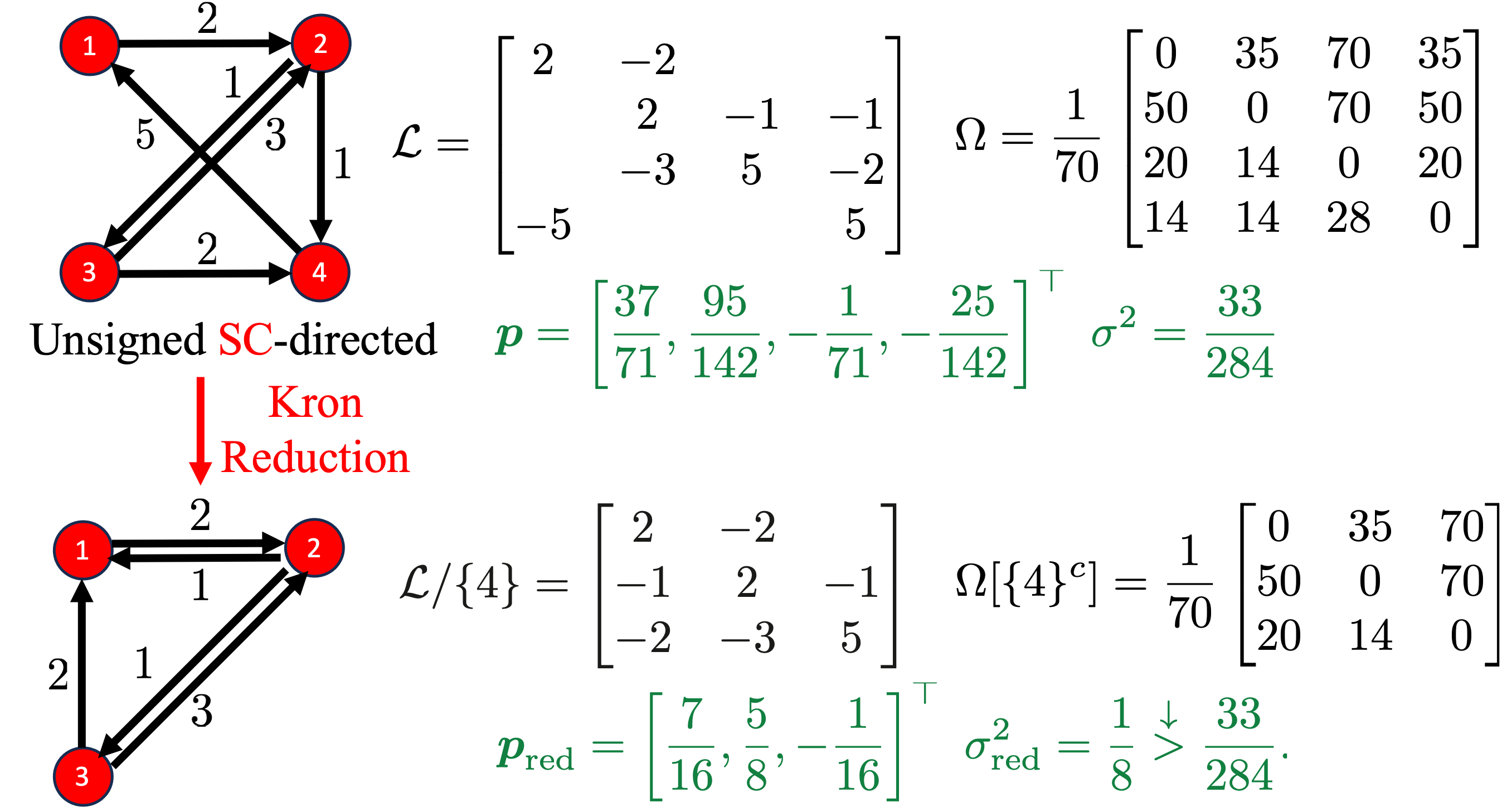}
  \caption{An example of an SC-directed, but not weight-balanced, graph for which Kron reduction increases the resistance radius.}
  \label{fig:example_sigma2_increase}
\end{figure}

\begin{example}[Failure of monotonicity of the resistance radius]
Consider the SC-directed, but not weight-balanced, graph shown in \cref{fig:example_sigma2_increase}. 
In this example, applying Kron reduction to node 4 increases the resistance radius. 
This illustrates that, in the general SC-directed case, the resistance radius is affected by the direction-dependent imbalances, so it is not necessarily a monotone size-related measure under Kron reduction.
\end{example}


This section is summarized in \cref{fig:commutative_diagram}: 
the commutativity of two algebraic operations on the SCWB-directed Laplacian; the calculation of resistance curvature and resistance radius in SC-directed graphs; the change in resistance curvature and resistance radius before and after Kron reduction; and the non-increasing property of the resistance radius in SCWB-directed graphs with respect to Kron reduction.

\begin{figure}[t]
  \centering
  \includegraphics[width=0.9\textwidth]{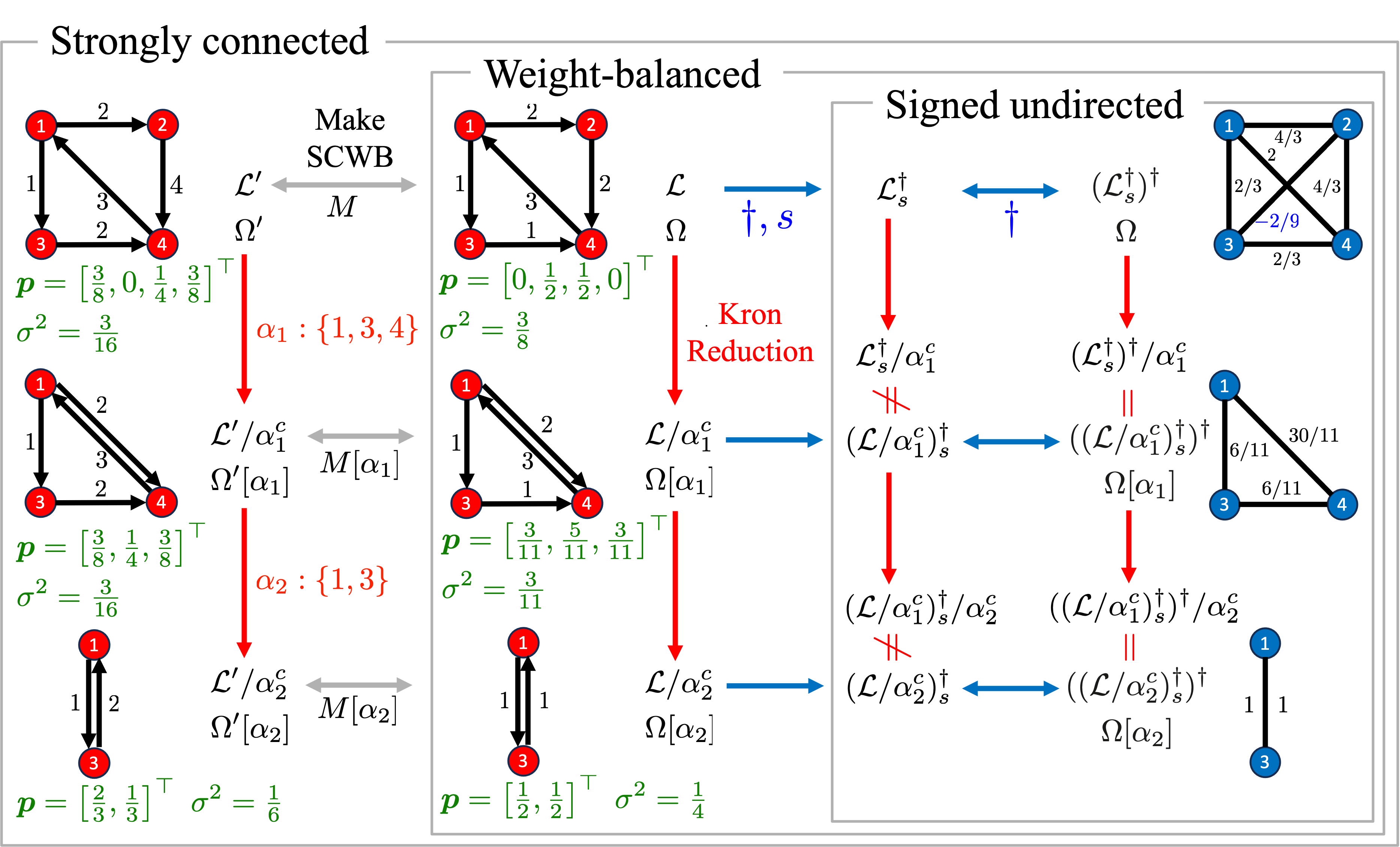}
  \caption{Summary of Section~3. The blue arrows represent the map $\mathcal{L} \mapsto (\mathcal{L}^{\dagger}_s)^{\dagger}$, and the red arrows represent Kron reduction. The gray arrows represent the process of making an SC-directed Laplacian weight-balanced. In the SCWB-directed case, these two operations commute. The figure also shows how the resistance curvature $\bm{p}$ and resistance radius $\sigma^2$ change under weight-balancing and Kron reduction.}
  \label{fig:commutative_diagram}
\end{figure}

\section{Extension of the resistance geometry framework}
\label{sec:extension_of_the_resistance_geometry_framework} 
We see that the matrix $(\mathcal{L}^{\dagger}_s)^{\dagger}$ constructed from an unsigned SCWB-directed Laplacian $\mathcal{L}$ has the Laplacian structure of a signed undirected graph (Thm.I\hspace{-1.2pt}V.4 of \cite{Fontan2023pseudoinverse}), that is, 
\begin{equation*}
    \begin{gathered}
        \mathcal{L}: \text{unsigned SCWB-directed Laplacian} \\
        \mapsto (\mathcal{L}^{\dagger}_s)^{\dagger} \succeq O,\ \ker( (\mathcal{L}^{\dagger}_s)^{\dagger}) = \operatorname{span}(\bm{1}): \text{signed undirected Laplacian}.
    \end{gathered}
\end{equation*}
Note also that the two graphs, $\mathcal{G}(\mathcal{L})$ and $\mathcal{G}((\mathcal{L}^{\dagger}_s)^{\dagger})$, have the same effective resistances, meaning that these resistances form a metric on the nodes in the latter signed graph as well. This motivates us to study resistance geometry of signed undirected Laplacians in relation to the spectral conditions of $(\mathcal{L}^{\dagger}_s)^{\dagger}$. 

Therefore, we consider the following signed undirected Laplacian classes:
\begin{align*}
    &
    \mathfrak{L}^{\pm}(n)
    := 
    \{\mathcal{Q} \in \mathbb{R}^{n \times n} : \mathcal{Q} = \mathcal{Q}^{\top},\ \ker({\mathcal{Q}}) \supseteq  \operatorname{span}(\bm{1}) \}, \\
    &
    \mathfrak{L}_{\mathrm{simplex}}^{\pm}(n)
    := 
    \{ \mathcal{Q} \in \mathfrak{L}^{\pm}(n) : \mathcal{Q} \succeq O,\ \ker(\mathcal{Q}) = \operatorname{span}(\bm{1}) \}.
\end{align*}

For every $\mathcal{Q}\in\mathfrak{L}^{\pm}_{\rm simplex}(n)$, each node of $\mathcal{G}(\mathcal{Q})$ has positive degree because $\mathcal{Q}_{ii} = \bm{1}_i^{\top} \mathcal{Q} \bm{1}_i  > 0 \ (\forall i \in [n])$. In this class, analogously to \cref{eq:ER_SCWB}, the effective resistance and the resistance matrix can be defined as follows:
\begin{equation}
    \label{eq:ER_for_Q_simplex}
    \begin{gathered}
        R_{\mathcal{G}}(a,b) := (\bm{1}_a-\bm{1}_b)^{\top}\mathcal{Q}^{\dagger}(\bm{1}_a-\bm{1}_b) 
        \qquad \forall\, a,b\in[n],\\
        \Omega := (R_{\mathcal{G}}(a,b))_{a,b\in[n]}
        = \bm{1}\bm{\zeta}^{\top}+\bm{\zeta}\bm{1}^{\top}-2\mathcal{Q}^{\dagger} \ \text{ where } \ \bm{\zeta}:=\operatorname{diagvec}(\mathcal{Q}^{\dagger}).
    \end{gathered}
\end{equation}
Under this definition, $R_{\mathcal{G}}$ is non-negative, symmetric, and satisfies $R_{\mathcal{G}}(a, b) = 0 \Leftrightarrow a = b$. Moreover, Kron reduction preserves the effective resistance (Lem.V\hspace{-1.2pt}I.3 of \cite{Fontan2023pseudoinverse}). 

In addition, based on the expressions in \cref{cor:FBI_generalized}, we can define the resistance curvature and resistance radius for any subset of nodes $\mathcal{V} \subseteq [n] \ (|\mathcal{V}| \geq 2)$ as
\begin{equation}
    \label{eq:resistance_curvature_radius_for_Q}
    \begin{aligned}
        \bm{p}(\mathcal{V}) := \frac{\Omega[\mathcal{V}, \mathcal{V}]^{-1} \bm{1}}{ \bm{1}^{\top} \Omega[\mathcal{V}, \mathcal{V}]^{-1} \bm{1}}, \quad
        \sigma^2(\mathcal{V}) := \frac{1}{2} (\bm{1}^{\top} \Omega[\mathcal{V}, \mathcal{V}]^{-1}\bm{1})^{-1}.
    \end{aligned}
\end{equation}
Note that the resistance matrix $\Omega$ and every principal submatrix $\Omega[\mathcal{V},\mathcal{V}]$ with $|\mathcal{V}|\geq 2$ are invertible.
Moreover, since $\mathcal{Q}/\mathcal{V}^c\in \mathfrak{L}_{\mathrm{simplex}}^{\pm}(|\mathcal{V}|)$ and its resistance matrix is $\Omega[\mathcal{V},\mathcal{V}]$, the same argument as in the proof of \cref{thm:FBI_generalized} yields $\bm{1}^{\top} \Omega[\mathcal{V},\mathcal{V}]^{-1} \bm{1} > 0$.

However, the effective resistance defined above \cref{eq:ER_for_Q_simplex} does not in general define a metric. For instance, consider the counterexample below:

\begin{example}
        \label{ex:Laplacian_Lsimplex_not_Lmetric}
        Consider the following signed undirected Laplacian in $\mathfrak{L}_{\mathrm{simplex}}^{\pm}(3)$:
        \begin{equation*}
            \mathcal{Q}
            =
            \frac{8}{63}
            \begin{bmatrix}
                8 & \mathbf{1} & -9 \\
                \mathbf{1} & 8 & -9 \\
                -9 & -9 & 18
            \end{bmatrix} \in  \mathfrak{L}_{\mathrm{simplex}}^{\pm}(3)
            \implies
            \Omega
            =
            \dfrac{1}{4}
            \begin{bmatrix}
                0 & 9 & 4 \\
                9 & 0 & 4 \\
                4 & 4 & 0
            \end{bmatrix}
        \end{equation*}
        In this example, $\mathcal{Q} \in \mathfrak{L}_{\mathrm{simplex}}^{\pm}(3)$ as its spectrum is $\operatorname{sp}(\mathcal{Q}) = \{0, 8/9, 24/7\}$. However, the resulting effective resistances, $R_{\mathcal{G}}(1,2) = 9/4,\ R_{\mathcal{G}}(1,3) = R_{\mathcal{G}}(2,3) = 1$, violate the triangle inequality: $R_{\mathcal{G}}(1,2) > R_{\mathcal{G}}(1,3) + R_{\mathcal{G}}(3, 2)$. Thus, the effective resistance in the class $\mathfrak{L}_{\mathrm{simplex}}^{\pm}(n)$ is not necessarily a metric between nodes.
\end{example}

Therefore, within the class $\mathfrak{L}_{\mathrm{simplex}}^{\pm}(n)$, we introduce a class in which the effective resistance serves as a metric between nodes, as follows.
\begin{align*}
    \mathfrak{L}_{\mathrm{metric}}^{\pm}(n)
    := 
    \{ \mathcal{Q} \in \mathfrak{L}_{\mathrm{simplex}}^{\pm}(n) : R_{\mathcal{G}}(a, b) + R_{\mathcal{G}}(b, c) \geq R_{\mathcal{G}}(a,c) \quad \forall a, b, c \in [n] \}.
\end{align*}
The effective resistance defined in the class $\mathfrak{L}_{\mathrm{metric}}^{\pm}(n)$ is a metric between nodes; hereafter, we refer to this metric as the \textit{generalized resistance metric}. This can be viewed as a generalized definition of the resistance metric for unsigned undirected graphs (Chaps.I\hspace{-1.2pt}V and V of \cite{Devriendt2022graph}). 

Moreover, Kron reduction preserves membership in $\{\mathfrak{L}_{\mathrm{metric}}^{\pm}(m)\}_{m \geq 2}$. Indeed, if $\mathcal{Q}\in\mathfrak{L}_{\mathrm{metric}}^{\pm}(n)$ and $\alpha\subset[n]$ with $|\alpha|\geq2$, then the resistance matrix induced by $\mathcal{Q}/\alpha^c$ is $\Omega[\alpha,\alpha]$. Since every principal submatrix of a metric matrix is again a metric matrix, we obtain $\mathcal{Q}/\alpha^c \in \mathfrak{L}_{\mathrm{metric}}^{\pm}(|\alpha|)$.

\begin{remark}[Geometric characterization of metric subclass]
    Let $\mathcal{Q} \in \mathfrak{L}_{\mathrm{simplex}}^{\pm}(n)$ and let $\mathcal{Q}^{\dagger}=B^{\top}B$, where
    $B=[\bm{b}_1\ \cdots\ \bm{b}_n] \in \mathbb{R}^{(n-1) \times n}$. Then, $R_{\mathcal{G}}(i, j) = \| \bm{b}_i - \bm{b}_j \|^2$ for all $a, b \in [n]$. Hence, for every $a,b,c \in[n]$,
    \begin{equation*}
        R_{\mathcal{G}}(i, j) + R_{\mathcal{G}}(j, k) - R_{\mathcal{G}}(i, k)
        =
        2(\bm{b}_i - \bm{b}_j)^{\top} (\bm{b}_k - \bm{b}_j).
    \end{equation*}
    Therefore,
    $\mathcal{Q}\in\mathfrak{L}_{\mathrm{metric}}^{\pm}(n)$ 
    if and only if every triangular face of the simplex generated by
    $\{\bm{b}_1,\ldots,\bm{b}_n\}$ is non-obtuse.
\end{remark}

Based on the preceding discussion, we have the following hierarchy of classes of signed undirected graph Laplacian matrices:
\begin{equation*}
    \mathfrak{L}^{\pm}(n) 
    \supsetneq 
    \mathfrak{L}_{\mathrm{simplex}}^{\pm}(n)
    \supsetneq 
    \mathfrak{L}_{\mathrm{metric}}^{\pm}(n)
    \ni 
    (\mathcal{L}^{\dagger}_s)^{\dagger} \mapsfrom \mathcal{L}: \text{unsigned SCWB-directed Laplacian}.
\end{equation*}

As reviewed in \cref{subsec:classification}, a metric matrix $D$ is of strict negative type if
$\bm{f}^{\top}D\bm{f}<0$ for all $\bm{f}\in\operatorname{span}(\bm{1})^{\perp} \backslash \{\bm{0}\}$.
We next show that the class of generalized resistance metrics $\mathfrak{L}_{\mathrm{metric}}^{\pm}(n)$ coincides with the class of strict negative type metrics.


\begin{theorem}
    [Strict negative type $\equiv$ generalized resistance metric]
    \label{thm:strict_negative_is_equivalent_to_generalized_resistance_metric}
    The class of strict negative type metric matrices coincides with the class
    of generalized resistance matrices induced by Laplacians in
    $\mathfrak{L}_{\mathrm{metric}}^{\pm}(n)$.
    More precisely,
    \begin{equation}
        \label{eq:strict_negative_generalized_resistance_equivalence}
        \begin{aligned}
            &
            \{ D=[d(a,b)]_{a,b\in[n]}: D \text{ is a strict negative type metric matrix} \}
            \\
            &\qquad =
            \{
                \Omega_{\mathrm{metric}}
                := \bm{1}\bm{\zeta}^{\top} + \bm{\zeta}\bm{1}^{\top} - 2\mathcal{Q}^{\dagger}
                :
                \mathcal{Q}\in\mathfrak{L}_{\mathrm{metric}}^{\pm}(n),\ \bm{\zeta} := \operatorname{diagvec}(\mathcal{Q}^{\dagger})
            \}.
        \end{aligned}
    \end{equation}
\end{theorem}

\begin{proof}
    We show the two inclusions separately.

    \noindent
    \textbf{(1) Strict negative type $\subseteq$ generalized resistance metric.}
    Let $D = [d(a,b)]$ be a strict negative type metric matrix and let $G$ be its corresponding centered Gram matrix.
    From \cref{eq:gram_euclid}, we obtain $\bm{f}^{\top}D\bm{f} = -2\bm{f}^{\top}G\bm{f} < 0$ for all $\bm{f} \in \operatorname{span}(\bm{1})^{\perp} \backslash \{\bm{0}\}$, and $G\bm{1}=\bm{0}$. In particular, $G$ is PSD with its kernel being $\operatorname{span}(\bm{1})$, so $G$ is a centered Gram matrix of a simplex (Prop.I\hspace{-1.2pt}V.15 of \cite{Devriendt2022graph}).

    Here, we define $\mathcal{Q}:=G^{\dagger}$. Hence, $\mathcal{Q}\in\mathfrak{L}_{\mathrm{simplex}}^{\pm}(n)$. Let $\bm{\xi}:=\operatorname{diagvec}(G)$. From \cref{eq:gram_euclid}, the distance matrix $D$ satisfies $D = \bm{1}\bm{\xi}^{\top} + \bm{\xi}\bm{1}^{\top} - 2G$. 
    Since $G = \mathcal{Q}^{\dagger}$ and $\bm{\xi}=\operatorname{diagvec}(\mathcal{Q}^{\dagger}) =\bm{\zeta}$, we obtain
    \begin{equation*}
        D = \bm{1}\bm{\zeta}^{\top} + \bm{\zeta}\bm{1}^{\top} - 2\mathcal{Q}^{\dagger}
        =
        \Omega_{\mathrm{metric}}.
    \end{equation*}

    Moreover, since $D$ is a metric matrix, its components satisfy the triangle inequality. Since $d(a,b) = R_{\mathcal{G}}(a,b)$, it follows that $R_{\mathcal{G}}(a,c) \leq R_{\mathcal{G}}(a,b) + R_{\mathcal{G}}(b,c)$ for all $a, b, c \in [n]$. 
    Therefore, we have $\mathcal{Q} \in \mathfrak{L}_{\mathrm{metric}}^{\pm}(n)$. This implies that every strict negative type metric matrix
    corresponds to a generalized resistance matrix induced by a Laplacian in $\mathfrak{L}_{\mathrm{metric}}^{\pm}(n)$.

    \noindent
    \textbf{(2) Generalized resistance metric
    $\subseteq$ strict negative type.}
    Conversely, let $\Omega_{\mathrm{metric}} = \bm{1}\bm{\zeta}^{\top} + \bm{\zeta}\bm{1}^{\top} - 2\mathcal{Q}^{\dagger}$
    be a resistance matrix induced by $\mathcal{Q} \in \mathfrak{L}_{\mathrm{metric}}^{\pm}(n)$, where $\bm{\zeta}:=\operatorname{diagvec}(\mathcal{Q}^{\dagger})$. 
    For any $\bm{f}\in\operatorname{span}(\bm{1})^{\perp} \backslash \{\bm{0}\}$, we have
    \begin{equation*}
        \bm{f}^{\top}\Omega_{\mathrm{metric}}\bm{f}
        = \bm{f}^{\top} ( \bm{1}\bm{\zeta}^{\top} + \bm{\zeta}\bm{1}^{\top} - 2\mathcal{Q}^{\dagger})\bm{f}
        = -2\bm{f}^{\top}\mathcal{Q}^{\dagger}\bm{f}
        < 0.
    \end{equation*}
    Thus, $\Omega_{\mathrm{metric}}$ is a strict negative type metric matrix.

    From (1) and (2), we conclude that the class of strict negative type metric matrices and the class of generalized resistance matrices
    induced by Laplacians in $\mathfrak{L}_{\mathrm{metric}}^{\pm}(n)$ are identical. This relationship can be expressed as follows:

    \[
        \begin{tikzcd}[column sep=normal,row sep=normal]
            \mathcal{Q}\in\mathfrak{L}_{\mathrm{metric}}^{\pm}(n)
              \arrow[r, "\dagger"]
              \arrow[d, equals]
            &
            \mathcal{Q}^{\dagger}
              \arrow[r]
              \arrow[d, equals]
            &
            \Omega_{\mathrm{metric}}
            =
            \bm{1}\bm{\zeta}^{\top}
            +
            \bm{\zeta}\bm{1}^{\top}
            -
            2\mathcal{Q}^{\dagger}
              \arrow[d, equals]
            \\
            G^{\dagger}
            &
            G
              \arrow[l, "\dagger"']
              \arrow[r]
            &
            D
            =
            \bm{1}\bm{\xi}^{\top}
            +
            \bm{\xi}\bm{1}^{\top}
            -
            2G
            :
            \text{strict negative type}
        \end{tikzcd}
    \]
\end{proof}

This theorem and the strict hierarchy of metric spaces can be summarized as 
\begin{align*}
    &\text{Negative} 
    \supsetneq 
    \text{Strict negative} \equiv \text{Generalized resistance} 
    \supsetneq 
    \text{Resistance}.
\end{align*}
This hierarchical structure enables us to combine results on Euclidean embeddings and graph-variance for the strict negative type metric spaces (Thm.I.2.4 of \cite{Fiedler2011} and Thm.I\hspace{-1.2pt}V.16, Prop.V.19 of \cite{Devriendt2022graph}) with our Laplacian-based resistance framework in $\mathfrak{L}_{\mathrm{metric}}^{\pm}(n)$.


\subsection{Maximum-variance problem under the generalized resistance metric}
\label{subsec:maxium_variance_problem_for_generalized_resistance_metric}
In the unsigned undirected setting, Devriendt \cite{Devriendt2022graph} showed that, when the underlying metric is the effective resistance, the resistance curvature characterizes the maximum variance distribution, while the resistance radius gives the optimal value.

Here, we consider the analogous problem for generalized resistance metrics induced by Laplacians in $\mathfrak{L}_{\mathrm{metric}}^{\pm}(n)$:
\begin{equation}
    \label{eq:maximum_varaance_prob_for_generalized_resistance_metric}
    \underset{\bm{f} \in \Delta_n}{\text{maximize}} \quad  \frac{1}{2} \bm{f}^{\top} \Omega_{\mathrm{metric}} \bm{f}.
\end{equation}
Since $\Omega_{\mathrm{metric}}$ is of strict negative type, the objective function is strictly concave on $\Delta_n$ and thus it has a unique maximizer $\bm{f}^* \in \Delta_n$ (Prop.V.19 of \cite{Devriendt2022graph}). 

We obtain a generalization of the corresponding result for the (classical) resistance metric discussed in Prop.V.21 of \cite{Devriendt2022graph}.

\begin{theorem}
\label{thm:maximum_variance_for_generalized_resistance_metric}
    Let $\bm{f}^*$ be the unique maximizer of the maximum graph-variance problem under the generalized resistance metric, and let $\mathcal{V}^*:=\operatorname{supp}(\bm{f}^*)$. Then $|\mathcal{V}^*|\geq 2$, and the following statements hold:
    \begin{enumerate}
        \item \textbf{Maximum variance:} The maximum variance is equal to the resistance radius of the maximum variance support: $\operatorname{var}^*_R = \sigma^2(\mathcal{V}^*)$.
        \item \textbf{Optimal distribution:} 
        The restriction of $\bm{f}^*$ to its support coincides with the resistance curvature of the Kron-reduced graph: $\bm{f}^{\ast}[\mathcal{V}^{\ast}] =\bm{p}(\mathcal{V}^{\ast})$.
    \end{enumerate}
\end{theorem}
\begin{proof}
    Let $\mathcal{Q} \in \mathfrak{L}_{\mathrm{metric}}^{\pm}(n)$, and let $\Omega$ be its resistance matrix. We first show that $|\mathcal{V}^*| \geq 2$. If $\bm{f} \in \Delta_n$ is supported on a single node, then $\bm{f}=\bm{1}_i$ for some $i\in[n]$, and hence $\frac{1}{2}\bm{f}^{\top}\Omega\bm{f} = \frac{1}{2} R_{\mathcal{G}}(i,i) = 0$. On the other hand, for any distinct $i,j\in[n]$, let $\hat{\bm{f}} := \frac{1}{2}(\bm{1}_i+\bm{1}_j) \in\Delta_n$. Then, $\frac{1}{2} \hat{\bm{f}}^{\top} \Omega \hat{\bm{f}} = \frac{1}{4}R_{\mathcal{G}}(i,j) > 0$. Therefore, no distribution supported on a single node can be optimal, and thus $|\mathcal{V}^*|\geq2$.
  
    Let $\Omega_{\mathrm{red}}$ denote the resistance matrix of the Kron-reduced graph $\mathcal{G}/(\mathcal{V}^*)^c$. From the necessary condition for $\bm{f}^*$ to be the maximum variance distribution in \cref{eq:necesaary_condition_2}, the vector $\bm{f}^*$ satisfies the following equation: 
    \begin{align*}
        &\Omega_{\text{red}} \bm{f}^*[\mathcal{V}^*]
        = \Omega [\mathcal{V}^*, \mathcal{V}^*] \bm{f}^*[\mathcal{V}^*] 
        = (\bm{1}^{\top} \Omega [\mathcal{V}^*, \mathcal{V}^*]^{\dagger} \bm{1})^{-1} \bm{1} \\
        &
        = (\bm{1}^{\top} \Omega [\mathcal{V}^*, \mathcal{V}^*]^{-1} \bm{1})^{-1} \bm{1} 
        \underset{\text{\cref{eq:resistance_curvature_radius_for_Q}}}{=} 2 \sigma^2 (\mathcal{V}^*) \bm{1}.
    \end{align*}
    The third equality is based on the fact that the resistance matrix is invertible. 
    Comparing this with $\Omega[\beta, \beta] \bm{p} (\beta) = 2 \sigma^2(\beta) \bm{1} \ (\beta \subseteq [n],\ |\beta| \geq 2)$ from  \cref{eq:resistance_curvature_radius_for_Q}, we have 
    \begin{gather*}
        \bm{f}^*[\mathcal{V}^*] = \bm{p}(\mathcal{V}^*)
        \ \text{and} \
        \bm{f}^*[(\mathcal{V}^*)^c] = \bm{0}, \\
        \operatorname{var}_R^* = \dfrac{1}{2}{\bm{f}^*}^{\top}\Omega \bm{f}^* = \frac{1}{2} \bm{p}(\mathcal{V}^*)^{\top} \Omega_{\text{red}} \bm{p}(\mathcal{V}^*) = \sigma^2(\mathcal{V}^*).
    \end{gather*}
    In other words, the restriction of $\bm{f}^*$ to its support corresponds to the resistance curvature in the Kron-reduced graph $\mathcal{G}(\mathcal{Q} / (\mathcal{V}^*)^c)$. In addition, the maximum variance value is equal to the resistance radius in that graph. 
\end{proof}

For connected unsigned undirected graphs, it is known that $p_i([n]) \geq f^*_i > 0 \ (\forall i \in \mathcal{V}^*)$ (Lem.V.22 of \cite{Devriendt2022graph}). This property relies on the nonnegativity of the edge weights and leads to a greedy algorithm for finding the exact maximizer, which often outperforms general-purpose quadratic programming; see Chap.V of \cite{Devriendt2022graph}.

However, for a genuinely signed Laplacian $\mathcal{Q}$ in $\mathfrak{L}_{\mathrm{metric}}^{\pm}(n)$, this property generally fails. Below, we provide a counterexample where this inequality fails.

\begin{example}
    \label{ex:Laplacian_Lmetric_not_satisfy}
    Consider the following signed undirected Laplacian on five nodes:
    \begin{equation*}
        \begin{gathered}
            \mathcal{Q} 
            =
            \frac{2}{609}
            \begin{bmatrix}
                 104 & -97 & \mathbf{8} & \mathbf{95} & -110 \\
                -97 & 272 & -148 & -235 & \mathbf{208} \\
                \mathbf{8} & -148 & 188 & \mathbf{101} & -149 \\
                \mathbf{95} & -235 & \mathbf{101} & 362 & -323 \\
                -110 & \mathbf{208} & -149 & -323 & 374
            \end{bmatrix}
            ,\
            \Omega =
            \begin{bmatrix}
                0 & 5 & 5 & 8 & 6 \\
                5 & 0 & 3 & 4 & 6 \\
                5 & 3 & 0 & 5 & 4 \\
                8 & 4 & 5 & 0 & 2 \\
                6 & 6 & 4 & 2 & 0
            \end{bmatrix}, \\
            \bm{p} ([5])
            =
            \frac{1}{203}
            \begin{bmatrix}
                88 \\
                -4 \\
                38 \\
                96 \\
                -15
            \end{bmatrix}
            ,\
            \sigma^2 ([5]) = \frac{424}{203}.
        \end{gathered}
    \end{equation*}
    The eigenvalues of $\mathcal{Q}$ are numerically given by $\operatorname{sp}(\mathcal{Q}) = \{0, 0.22, 0.35, 0.55, 3.16\}$, i.e., $\mathcal{Q} \in \mathfrak{L}_{\mathrm{simplex}}^{\pm}(5)$. A direct inspection of the resistance matrix shows that all triangle inequalities are satisfied. Hence, $\mathcal{Q}\in\mathfrak{L}_{\mathrm{metric}}^{\pm}(5)$. Solving the quadratic optimization problem yields the unique maximizer:
    \begin{equation*} 
        \bm{f}^*
        = 
        \frac{1}{374}
        \begin{bmatrix}
            154 \\
            8 \\
            59 \\
            153 \\ 
            0
        \end{bmatrix}
        ,\ 
        \operatorname{var}^*_R = \frac{1559}{748}.
    \end{equation*}
    The restriction of $\bm{f}^*$ to its support satisfies $ \bm{f}^*[\{5\}^c] = \bm{p}(\{5\}^c)$, where $\bm{p}(\{5\}^c)$ is the resistance curvature of the Kron-reduced graph $\mathcal{G} / \{5\}$. Moreover, $\operatorname{var}^*_R = \sigma^2(\{5\}^c)$.

    However, for node 2, we observe that $f^*_2 = p_2(\{5\}^c) = 4/187 > 0$ despite $p_2 ([5]) = -4/203 < 0$. This cannot occur for connected unsigned undirected Laplacians. 
\end{example}


\begin{remark}
    We note that the characterization in \cref{thm:maximum_variance_for_generalized_resistance_metric}
    remains valid for $\mathcal{Q} \in \mathfrak{L}_{\mathrm{simplex}}^{\pm}(n)$, since its proof does not use the triangle inequality. If
    $\mathcal{Q} \in \mathfrak{L}_{\mathrm{simplex}}^{\pm}(n) \backslash \mathfrak{L}_{\mathrm{metric}}^{\pm}(n)$, however, the resulting effective resistance is not a metric, so the objective can no longer be interpreted as graph-variance induced by a node metric. 
\end{remark}


\subsection{Euclidean simplex realization and generalized resistive embedding}
\label{subsec:generalized_resistive_embedding}

\textit{Resistive embedding} is a method of realizing the $n$ nodes of a graph as $n$ points in $\mathbb{R}^{n-1}$ such that the squared Euclidean distance between each pair of points equals the corresponding effective resistance. For the unsigned undirected setting, see Chap.~IV of \cite{Devriendt2022graph}.

In this subsection, we first consider the Euclidean simplex associated with a Laplacian $\mathcal{Q}\in\mathfrak{L}^{\pm}_{\mathrm{simplex}}(n)$. The construction requires only that $\mathcal{Q}$ is PSD with its kernel being $\operatorname{span}(\bm{1})$; it does not require the resulting effective resistances to satisfy the triangle inequality. For every $\mathcal{Q} \in \mathfrak{L}^{\pm}_{\mathrm{simplex}}(n)$, the effective resistances can be realized as the squared Euclidean distances between the vertices of a centered simplex.

When $\mathcal{Q} \in \mathfrak{L}^{\pm}_{\mathrm{metric}}(n)$, the effective resistances additionally form a metric on the node set, and we refer to the resulting realization as a \textit{generalized resistive embedding}. When $\mathcal{Q} \in \mathfrak{L}^{\pm}_{\mathrm{simplex}}(n) \backslash \mathfrak{L}^{\pm}_{\mathrm{metric}}(n)$, the same construction is instead interpreted as a squared Euclidean simplex realization, rather than as an embedding of a node metric.

Let $\lambda_1\geq\cdots\geq\lambda_{n-1}>\lambda_n=0$ be the eigenvalues of $\mathcal{Q}$, with corresponding orthonormal eigenvectors $\bm{v}_1,\ldots,\bm{v}_{n-1},\bm{v}_n=\bm{1}/\sqrt{n}$. The pseudoinverse of $\mathcal{Q}$ is given by $\mathcal{Q}^{\dagger}=\sum_{k=1}^{n-1} \lambda_k^{-1} \bm{v}_k \bm{v}_k^{\top}$. For each $i \in [n]$, define $\bm{b}_i \in \mathbb{R}^{n-1}$ by $(\bm{b}_i)_k:=\lambda_k^{-1/2}(\bm{v}_k)_i$ for all $k \in [n-1]$. Then, for all $i,j\in[n]$, we have 
\begin{equation}
    \label{eq:generalized_resistive_embedding_property_1}
    \bm{b}_i^{\top} \bm{b}_j
    = \sum_{k=1}^{n-1} \lambda_k^{-1} (\bm{v}_k)_i (\bm{v}_k)_j
    = \mathcal{Q}^{\dagger}_{ij}, \quad
    \|\bm{b}_i - \bm{b}_j \|^2
    = \mathcal{Q}^{\dagger}_{ii} + \mathcal{Q}^{\dagger}_{jj} - 2 \mathcal{Q}^{\dagger}_{ij}
    \underset{\cref{eq:ER_for_Q_simplex}}{=} R_{\mathcal{G}}(i, j).
    \end{equation}
It follows that $\mathcal{Q}^{\dagger} = B^{\top} B$, where $B= [\bm{b}_1 \cdots \bm{b}_n ]$. Combining this with $\ker(\mathcal{Q}^{\dagger}) = \operatorname{span}(\bm{1})$, we see that $\mathcal{Q}^{\dagger}$ is a centered Gram matrix of a simplex (Prop.I\hspace{-1.2pt}V.15 of \cite{Devriendt2022graph}). 
Since $\| B \bm{1} \|^2 = \bm{1}^{\top} \mathcal{Q}^{\dagger} \bm{1} = 0$, we know that $B \bm{1} = \bm{0}$\footnote{
    Although the factorization $\mathcal{Q}^{\dagger} = B^{\top} B$ is not unique, every such full-row-rank factorization satisfies $B \bm{1} = \bm{0}$. The corresponding simplices therefore differ only by an orthogonal transformation and have their centroid at the origin.
}.
Assume that there exists a vector $\bm{f} \in \mathbb{R}^n \setminus \{ \bm{0}\}$ such that $\bm{f}^{\top} \bm{1} = 0$ and $B\bm{f} = \bm{0}$. 
Then, $\| B \bm{f} \|^2  = \bm{f}^{\top} \mathcal{Q}^{\dagger} \bm{f} = 0$, which contradicts the fact that 
$\ker(\mathcal{Q}^{\dagger}) = \operatorname{span}(\bm{1})$. Hence, $\ker(B) = \operatorname{span}(\bm{1})$. 
Let $\bm{u}_i := B \bm{v}_i / \| B \bm{v}_i \| = \lambda_i^{1/2} B \bm{v}_i \ (\forall i \in [n-1])$. 
Then, we have $\bm{u}_i^{\top} \bm{u}_j = \delta_{ij} \ (i, j \in [n-1])$ and 
\begin{gather}
    \begin{aligned}
        \label{eq:SVD_B_Bdagger}
        B &= \sum_{i=1}^{n-1} \lambda_i^{-1/2}\,\bm u_i \bm v_i^{\top}, 
        \quad
        B^{\dagger} = \sum_{i=1}^{n-1} \lambda_i^{1/2}\,\bm v_i \bm u_i^{\top},
    \end{aligned}\\
    \label{eq:SVD_B_Bdagger_2}
    \begin{aligned}
        B^{\dagger}(B^{\dagger})^{\top}
        &= \sum_{i=1}^{n-1} \lambda_i\,\bm v_i \bm v_i^{\top}
        = \mathcal{Q}
        = (B^{\top}B)^{\dagger}.
    \end{aligned}
\end{gather}

Based on the above, we define the associated Euclidean simplex realization.

\begin{definition}[Euclidean simplex realization and generalized resistive embedding]
    \label{def:generalized_resistive_embedding}
    Let $\mathcal{Q} \in \mathfrak{L}^{\pm}_{\mathrm{simplex}}(n)$, and let $\mathcal{Q}^{\dagger} = B^{\top}B$.
    Define $\varphi: \{\bm{f}\in\mathbb{R}^n:\bm{1}^{\top}\bm{f}=1\} \to \mathbb{R}^{n-1}$ by
    $\varphi(\bm{f}):=B\bm{f}$. We call $\varphi$ the
    \textit{Euclidean simplex realization} associated with $\mathcal{Q}$. If $\mathcal{Q} \in \mathfrak{L}^{\pm}_{\mathrm{metric}}(n)$, we call
    $\varphi$ the \textit{generalized resistive embedding}.
\end{definition}

This construction extends the resistive embedding for connected unsigned undirected graphs (Def.I\hspace{-1.2pt}V.25 of \cite{Devriendt2022graph}). For any $i,j\in[n]$, \cref{eq:generalized_resistive_embedding_property_1} implies $\|\varphi(\bm{1}_i) -\varphi(\bm{1}_j)\|^2 =R_{\mathcal{G}}(i,j)$. The points $\{\varphi(\bm{1}_1),\ldots, \varphi(\bm{1}_n)\}$ form a simplex, denoted by $\mathcal{S}$. 
In other words, 
\begin{equation*}
      \mathcal{S} 
      := \operatorname{conv}\{ \varphi(\bm{1}_1), \dots, \varphi(\bm{1}_n)\} 
      = \left\{ \textstyle\sum_{i=1}^n g_i \varphi(\bm{1}_i) : \bm{g}^{\top} \bm{1} = 1,\ \bm{0} \leq \bm{g} \in \mathbb{R}^n \right\}.
\end{equation*}
In addition, for a subset of nodes $\mathcal{V} \subsetneq [n],\ |\mathcal{V}|\geq 2$, the $\mathcal{V}$-face of $\mathcal{S}$ is defined as follows (see \cite{Fiedler2011_2} and Chap.I\hspace{-1.2pt}V of \cite{Devriendt2022graph}):
\begin{equation*}
    \mathcal{V}\text{-face} 
    := \operatorname{conv}\{ \varphi(\bm{1}_i) : i \in \mathcal{V}\} 
    = \left\{ \textstyle\sum_{i \in \mathcal{V}} g_i \varphi (\bm{1}_i) : \bm{g}^{\top} \bm{1} = 1,\ \bm{0} \leq \bm{g} \in \mathbb{R}^{|\mathcal{V}|} \right\}.
\end{equation*}
Moreover, for the center of gravity of $\mathcal{S}$, we have
\begin{equation}
    \label{eq:generalized_resistive_embedding_1}
    \varphi(\bm{1}/n) = B \bm{1}/n = \bm{0}, \quad \|\varphi(\bm{1}_i) \|^2 
    = \|B \bm{1}_i \|^2 = \bm{1}_i^{\top} \mathcal{Q}^{\dagger} \bm{1}_i = \mathcal{Q}^{\dagger}_{ii} > 0. 
\end{equation}
This states that the squared distance between the center of gravity $\varphi(\bm{1}/n)$ of the simplex, which is the origin, and each vertex $\varphi(\bm{1}_i)$ is expressed by the corresponding diagonal element of the pseudoinverse Laplacian.

Next, we introduce the geometric properties of $\bm{p}$ and $\sigma^2$ through the Euclidean simplex realization.

\begin{proposition}[circumsphere interpretation, a generalization of Prop.I\hspace{-1.2pt}V.32 in \cite{Devriendt2022graph}]
     Consider a signed undirected graph specified by Laplacian $\mathcal{Q}$ in $\mathfrak{L}_{\mathrm{simplex}}^{\pm}(n)$. 
     The circumcenter of the simplex $\mathcal{S}$ is $\varphi(\bm{p})$, and its circumradius is equal to $\sigma$. 
\end{proposition}
\begin{proof}
    The squared distance between $\varphi(\bm{p})$ and each vertex $\varphi(\bm{1}_i)$ is
    \begin{align*}
        &\|\varphi(\bm{p}) - \varphi(\bm{1}_i)\|^2 
        = \|B (\bm{p} - \bm{1}_i) \|^2 
        = (\bm{p} - \bm{1}_i)^{\top} \mathcal{Q}^{\dagger} (\bm{p} - \bm{1}_i) \\
        &\underset{\cref{eq:ER_for_Q_simplex}}{=} -\frac{1}{2} \bm{p}^{\top} \Omega \bm{p} + \bm{1}_i^{\top} \Omega \bm{p} 
        \underset{\cref{eq:resistance_curvature_radius_for_Q}}{=} -\sigma^2 + \bm{1}_i^{\top} (2\sigma^2 \bm{1}) 
        =\sigma^2.
    \end{align*}
    Since this holds for any $i \in [n]$, the distances from $\varphi(\bm{p})$ to each vertex $\varphi(\bm{1}_i)$ of $\mathcal{S}$ are all equal and $\sigma > 0$. In other words, the circumcenter of $\mathcal{S}$ is $\varphi(\bm{p})$, and its circumradius is $\sigma$.
\end{proof}

The next proposition gives a geometric interpretation of the entries of $\mathcal{Q}$ and $\mathcal{Q}^\dagger$. Since any Laplacian $\mathcal{Q}$ in $\mathfrak{L}_{\mathrm{simplex}}^{\pm}(n)$ defines a Euclidean simplex via $\mathcal{Q}^{\dagger}$, Fiedler's classical angle formulas (Chap.I\hspace{-1.2pt}I\hspace{-1.2pt}I of \cite{Fiedler2011_2}) apply to our setting, even when the underlying graph has signed edges.

\begin{proposition}[Fiedler-type angle formulas \cite{Fiedler2011_2}]
    \label{prop:angles_fiedler_type}
    Let $n \geq 3$, and let $\mathcal{Q}\in\mathfrak{L}^{\pm}_{\mathrm{simplex}}(n)$ be a signed undirected Laplacian.
    Let $\theta_{ij}$ denote the interior dihedral angle between the $i^c$-face and $j^c$-face of $\mathcal{S}$, and let $\phi_{ij}$ be the angle subtended by the vertices $\varphi(\bm{1}_i)$ and $\varphi(\bm{1}_j)$ at the center of gravity of $\mathcal{S}$. 
    Then, the following formulas hold: 
    \begin{equation*}
        \cos\theta_{ij} = -\frac{\mathcal{Q}_{ij}}{\sqrt{\mathcal{Q}_{ii} \mathcal{Q}_{jj}}},\quad
        \cos\phi_{ij} = \frac{\mathcal{Q}^{\dagger}_{ij}}{\sqrt{\mathcal{Q}^{\dagger}_{ii} \mathcal{Q}^{\dagger}_{jj}}}.
    \end{equation*}
\end{proposition}

\begin{proof}
    The linear subspace parallel to the $i^c$-face is $\mathcal{H}_i := \{B\bm{f}:\bm{f}^{\top}\bm{1}=0,\ f_i=0\}$.
    Indeed, $\mathcal{H}_i$ is generated by the differences $\bm{b}_j-\bm{b}_k$ with $j,k\neq i$.
    Thus, for any $\bm{f}\in\mathbb{R}^n$ such that $\bm{f}^{\top}\bm{1}=0$ and $f_i=0$, the vector $B\bm{f}$ lies in $\mathcal{H}_i$, the direction space parallel to the $i^c$-face. Then,
    \begin{equation*}
        ((B^{\dagger})^{\top} \bm{1}_i )^{\top} B \bm{f} 
        = \bm{1}_i^{\top} B^{\dagger}B \bm{f}
        \underset{\cref{eq:SVD_B_Bdagger}}{=} \bm{1}_i^{\top}(I_n - \bm{1}\bm{1}^{\top} / n) \bm{f}
        =0.
    \end{equation*}
    Therefore, $(B^{\dagger})^{\top} \bm{1}_i$ is a normal vector of the $i^c$-face. Since $((B^{\dagger})^{\top} \bm{1}_i )^{\top} \bm{b}_i = 1 -1/n >0$ and $((B^{\dagger})^{\top}\bm{1}_i)^{\top}\bm{b}_j = -1/n \ (\forall j\neq i)$ from \cref{eq:SVD_B_Bdagger}, $(B^{\dagger})^{\top} \bm{1}_i$ points to the same side of the $i^c$-face as the vertex $\bm{b}_i $. 
    As illustrated in \cref{fig:picture_for_angles}, the dihedral angle $\theta_{ij}$ can be calculated from the angle between the two vectors $(B^{\dagger})^{\top} \bm{1}_i$ and $(B^{\dagger})^{\top} \bm{1}_j$ as follows:
    \begin{equation*}
        \cos(\pi-\theta_{ij})
        = 
        -\cos\theta_{ij}
        = 
        \frac{ \bm{1}_i^{\top} B^{\dagger} (B^{\dagger})^{\top}\bm{1}_j } {\|(B^{\dagger})^{\top}\bm{1}_i \| \|(B^{\dagger})^{\top}\bm{1}_j \|} 
        \underset{\cref{eq:SVD_B_Bdagger_2}}{=}
        \frac{\bm{1}_i^{\top}\mathcal{Q}\bm{1}_j}{\sqrt{\bm{1}_i^{\top}\mathcal{Q}\bm{1}_i  \bm{1}_j^{\top}\mathcal{Q}\bm{1}_j}} 
        = 
        \frac{\mathcal{Q}_{ij}}{\sqrt{\mathcal{Q}_{ii}\mathcal{Q}_{jj}}}.
    \end{equation*}

    Next, we calculate $\phi_{ij}$. 
    Applying the cosine rule to the triangle formed by the center of gravity and the two points $\varphi(\bm{1}_i)$ and $\varphi(\bm{1}_j)$, we have 
    \begin{equation*} 
        \cos{\phi_{ij}} 
        = 
        \frac{\| \varphi(\bm{1}_i) \|^2 + \| \varphi(\bm{1}_j) \|^2 - R_{\mathcal{G}}(i, j)}{2 \| \varphi(\bm{1}_i) \| \| \varphi(\bm{1}_j) \|}
        \underset{\cref{eq:generalized_resistive_embedding_1}}{=} 
        \frac{\mathcal{Q}^{\dagger}_{ii} + \mathcal{Q}^{\dagger}_{jj} - R_{\mathcal{G}}(i, j)}{2\sqrt{\mathcal{Q}^{\dagger}_{ii}  \mathcal{Q}^{\dagger}_{jj}}}
        \underset{\cref{eq:ER_for_Q_simplex}}{=} 
        \frac{\mathcal{Q}^{\dagger}_{ij}}{\sqrt{\mathcal{Q}^{\dagger}_{ii}  \mathcal{Q}^{\dagger}_{jj}}}.
    \end{equation*}
    This completes the proof.
\end{proof}

\begin{figure}[t]
    \centering
    \label{fig:picture_for_angles}
    \includegraphics[width = 0.6\textwidth]{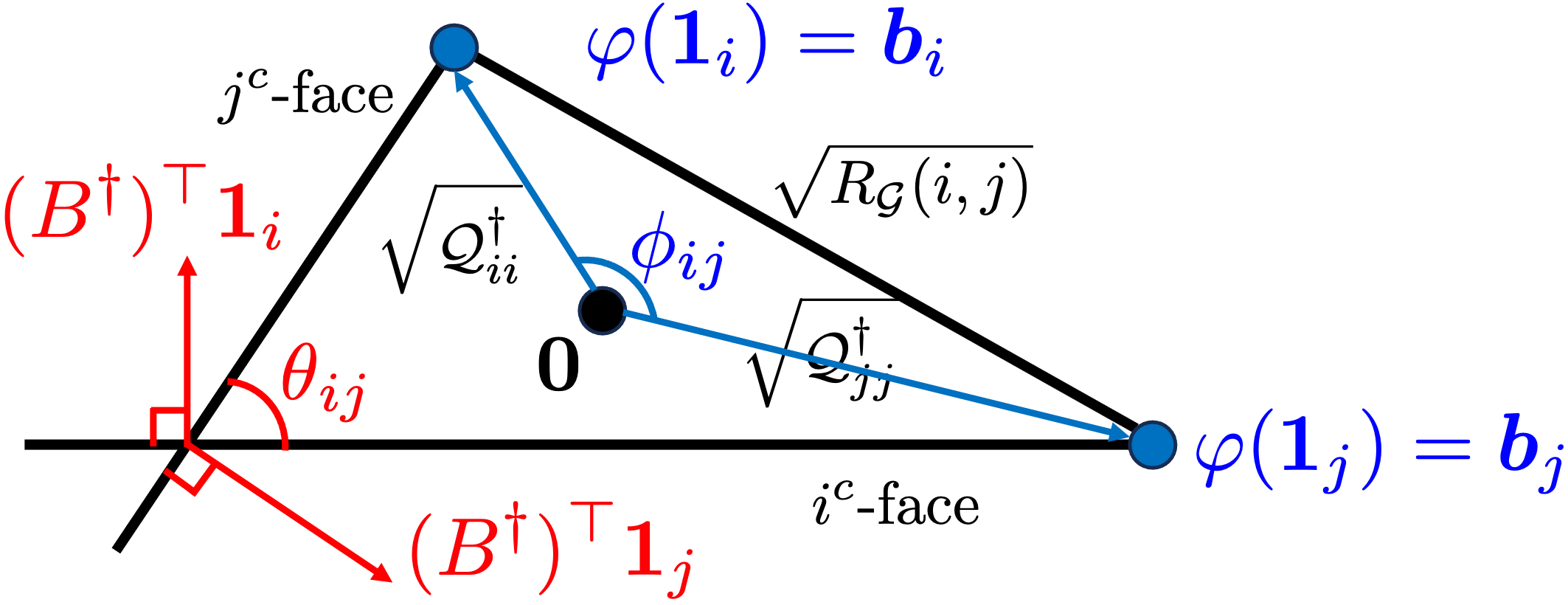}
    \caption{The inward normal vectors of the $i^c$- and $j^c$-faces are given by $(B^\dagger)^\top\bm{1}_i$ and $(B^\dagger)^\top\bm{1}_j$, respectively. The angle between these normals is $\pi-\theta_{ij}$, where $\theta_{ij}$ is the interior dihedral angle between the two faces.}
\end{figure}

\cref{prop:angles_fiedler_type} is classical, goes back to Fiedler’s simplex interpretation of Laplacians \cite{Fiedler2011_2}, and has been revisited in the modern Laplacian--simplex literature. The contribution here is that the same geometric correspondence applies to any matrix in $\mathfrak{L}_{\mathrm{simplex}}^{\pm}(n)$ arising from our generalized resistance framework, even in the presence of signed edges. 

In this interpretation, the sign pattern of $\mathcal{Q}$ directly reflects the geometry of the associated simplex: positive off-diagonal entries correspond to
obtuse dihedral angles between faces, while $\mathcal{Q}^\dagger$ encodes the pairwise angles between embedded vertices as seen from the center of gravity.
See \cref{fig:example_generalized_resistive_embedding} for a concrete example of this correspondence.

\begin{figure}[t]
    \centering
    \includegraphics[width = \textwidth]{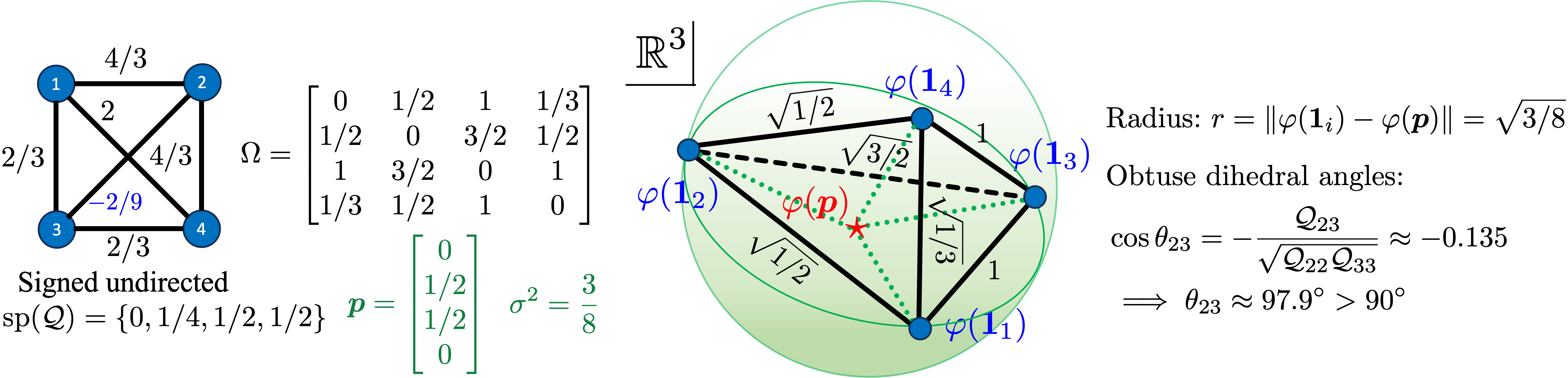}
    \caption{An example of the generalized resistive embedding in the metric subclass $\mathfrak{L}^{\pm}_{\mathrm{metric}}(4)$. The set of vertices, $\{\varphi(\bm{1}_1), \dots, \varphi(\bm{1}_4) \}$, forms a simplex, and the squared distances between pairs of vertices equal the corresponding effective resistances. The circumcenter of the simplex is $\varphi(\bm{p})$, and its circumradius is $\sigma = \sqrt{3/8}$. Since $\mathcal{Q}_{23} > 0$, the angle between the $2^c$-face and the $3^c$-face is obtuse.}
    \label{fig:example_generalized_resistive_embedding}
\end{figure}

\section{Conclusions}
\label{sec:conclusion}

Building on Kron reduction and effective resistance, we developed a generalized resistance geometry for SCWB-directed graphs and signed
undirected graphs. We established a generalized Fiedler--Bapat identity for SCWB-directed graphs, providing a foundation for defining the resistance curvature $\bm{p}$ and resistance radius $\sigma^2$. We also derived the directionality decomposition $(\mathcal{L}^{\dagger}_s)^{\dagger} = \mathcal{L}_s + \mathcal{K}(\mathcal{L}_s)^{\dagger}\mathcal{K}^{\top}$, and showed that it yields a pairwise effective-resistance comparison between
the associated signed undirected graph and the symmetrized graph. In the unsigned setting, this comparison extends a previous result obtained
under a normality assumption to all SCWB-directed Laplacians. Furthermore, we proved that the undirecting map $\mathcal{L} \mapsto (\mathcal{L}^{\dagger}_s)^{\dagger}$ commutes with Kron reduction. 
We introduced generalized resistance metrics for signed undirected graphs and showed that their resistance matrices are precisely the strict negative type metrics. Moreover, we characterized the unique optimaizer of maximum graph-variance problem, and developed generalized resistive embeddings. 

Future work includes extending the framework to broader classes of directed and signed Laplacians and clarifying the conditions under which $\bm{p}$ and $\sigma^2$ are well-defined for general SC-directed graphs. Another direction is graph comparison based on the proposed geometric invariants. 
For graphs with the same labeled node set, one may compare their resistance curvatures using $\|\bm{p}-\bm{p}'\|_2$ whenever both resistance curvatures are well-defined. For unlabeled graphs, node permutations should also be considered, while graphs with different numbers of nodes may be compared after Kron reduction onto suitable node subsets.

\bibliographystyle{siamplain}
\bibliography{Generalized_resistance_geometry_references}

\end{document}

%% file: Generalized_resistance_geometry_shared.tex
\usepackage{lipsum}
\usepackage{amsfonts}
\usepackage{graphicx}
\usepackage{epstopdf}
\usepackage{algorithmic}
\ifpdf
  \DeclareGraphicsExtensions{.eps,.pdf,.png,.jpg}
\else
  \DeclareGraphicsExtensions{.eps}
\fi

\newsiamremark{remark}{Remark}
\newsiamremark{hypothesis}{Hypothesis}
\crefname{hypothesis}{Hypothesis}{Hypotheses}
\newsiamthm{claim}{Claim}

\headers{Generalized Resistance Geometry from KR and ER}{Y. Kajiura and K. Sato}

\title{Generalized Resistance Geometry from Kron reduction and Effective Resistance
    \thanks{Submitted to the editors DATE.\funding{
    This work was supported by Japan Society for the Promotion of Science KAKENHI under 26K03232.}
    }
}

\author{%
  Yosuke Kajiura%
  \and
  Kazuhiro Sato%
  \thanks{%
    Department of Mathematical Informatics, Graduate School of Information Science and Technology,
    The University of Tokyo, 7-3-1 Hongo, Bunkyo-ku, Tokyo 113-8656, Japan
    (\email{kyh140207@g.ecc.u-tokyo.ac.jp}, \email{kazuhiro@mist.i.u-tokyo.ac.jp}).%
  }%
}

\usepackage{amsopn}
\DeclareMathOperator{\diag}{diag}

\makeatletter
\newcommand*{\addFileDependency}[1]{
  \typeout{(#1)}
  \@addtofilelist{#1}
  \IfFileExists{#1}{}{\typeout{No file #1.}}
}
\makeatother

\newcommand*{\myexternaldocument}[1]{%
    \externaldocument{#1}%
    \addFileDependency{#1.tex}%
    \addFileDependency{#1.aux}%
}

%% file: Generalized_resistance_geometry_references.bib
@inproceedings{Altafini2019Investigating,
  author    = {Altafini, Claudio},
  title     = {Investigating Stability of {L}aplacians on Signed Digraphs via Eventual Positivity},
  booktitle = {2019 IEEE 58th Conference on Decision and Control (CDC)},
  year      = {2019},
  pages     = {5044--5049},
  doi       = {10.1109/CDC40024.2019.9030125}
}

@book{Blumenthal1953theory,
  author    = {Blumenthal, Leonard M.},
  title     = {Theory and Applications of Distance Geometry},
  publisher = {Clarendon Press},
  address   = {Oxford},
  year      = {1953}
}

@article{Dawkins2024NodeRC,
  author  = {Dawkins, Aleyah and
             Gupta, Vishal and
             Kempton, Mark and
             Linz, William and
             Quail, Jeremy and
             Richman, D. Harry and
             Stier, Zachary},
  title   = {Node Resistance Curvature in {C}artesian Graph Products},
  journal = {Journal of Combinatorics},
  volume  = {16},
  number  = {4},
  pages   = {575--589},
  year    = {2025},
  doi     = {10.4310/JOC.250923003413}
}

@phdthesis{Devriendt2022graph,
  author = {Devriendt, Karel},
  title  = {Graph Geometry from Effective Resistances},
  school = {University of Oxford},
  type   = {{DPhil} thesis},
  year   = {2022},
  doi    = {10.5287/ora-j1p4rornb}
}

@article{Devriendt2024resistancedistance,
  author  = {Devriendt, Karel and
             Ottolini, Andrea and
             Steinerberger, Stefan},
  title   = {Graph Curvature via Resistance Distance},
  journal = {Discrete Applied Mathematics},
  volume  = {348},
  pages   = {68--78},
  year    = {2024},
  doi     = {10.1016/j.dam.2024.01.012}
}

@book{Doyle_Snell_1984,
  author    = {Doyle, Peter G. and Snell, J. Laurie},
  title     = {Random Walks and Electric Networks},
  publisher = {Mathematical Association of America},
  series    = {Carus Mathematical Monographs},
  volume    = {22},
  address   = {Washington, DC},
  year      = {1984}
}

@article{Dorfler2013,
  author  = {D{\"o}rfler, Florian and Bullo, Francesco},
  title   = {{K}ron Reduction of Graphs with Applications to Electrical Networks},
  journal = {IEEE Transactions on Circuits and Systems I: Regular Papers},
  volume  = {60},
  number  = {1},
  pages   = {150--163},
  year    = {2013},
  doi     = {10.1109/TCSI.2012.2215780}
}

@article{Fitch2019EffectiveResistance,
  author  = {Fitch, Katherine},
  title   = {Effective Resistance Preserving Directed Graph Symmetrization},
  journal = {SIAM Journal on Matrix Analysis and Applications},
  volume  = {40},
  number  = {1},
  pages   = {49--65},
  year    = {2019},
  doi     = {10.1137/18M1172892}
}

@article{Fiedler1962,
  author  = {Fiedler, Miroslav and Pt{\'a}k, Vlastimil},
  title   = {On Matrices with Non-Positive Off-Diagonal Elements and Positive Principal Minors},
  journal = {Czechoslovak Mathematical Journal},
  volume  = {12},
  number  = {3},
  pages   = {382--400},
  year    = {1962},
  doi     = {10.21136/CMJ.1962.100526}
}

@incollection{Fiedler2011,
  author    = {Fiedler, Miroslav},
  title     = {Special Simplexes},
  booktitle = {Matrices and Graphs in Geometry},
  publisher = {Cambridge University Press},
  series    = {Encyclopedia of Mathematics and its Applications},
  volume    = {139},
  pages     = {64--113},
  address   = {Cambridge},
  year      = {2011},
  isbn      = {978-0-521-46193-1}
}

@book{Fiedler2011_2,
  author    = {Fiedler, Miroslav},
  title     = {Matrices and Graphs in Geometry},
  publisher = {Cambridge University Press},
  series    = {Encyclopedia of Mathematics and its Applications},
  volume    = {139},
  address   = {Cambridge},
  year      = {2011},
  isbn      = {978-0-521-46193-1}
}

@inproceedings{Fontan2021pseudoinverse,
  author    = {Fontan, Angela and Altafini, Claudio},
  title     = {On the Properties of {L}aplacian Pseudoinverses},
  booktitle = {2021 60th IEEE Conference on Decision and Control (CDC)},
  year      = {2021},
  pages     = {5538--5543},
  doi       = {10.1109/CDC45484.2021.9683525}
}

@article{Fontan2023pseudoinverse,
  author  = {Fontan, Angela and Altafini, Claudio},
  title   = {Pseudoinverses of Signed {L}aplacian Matrices},
  journal = {SIAM Journal on Matrix Analysis and Applications},
  volume  = {44},
  number  = {2},
  pages   = {622--647},
  year    = {2023},
  doi     = {10.1137/22M1493392}
}

@book{GodsilRoyle2001AGT,
  author    = {Godsil, Chris and Royle, Gordon},
  title     = {Algebraic Graph Theory},
  publisher = {Springer},
  series    = {Graduate Texts in Mathematics},
  volume    = {207},
  address   = {New York, NY},
  year      = {2001},
  doi       = {10.1007/978-1-4613-0163-9},
  isbn      = {978-0-387-95241-3}
}

@book{Kron1939TAN,
  author    = {Kron, Gabriel},
  title     = {Tensor Analysis of Networks},
  publisher = {John Wiley \& Sons and Chapman \& Hall},
  address   = {New York and London},
  series    = {General Electric Series},
  year      = {1939}
}

@article{Klein1993,
  author  = {Klein, Douglas J. and Randi{\'c}, Milan},
  title   = {Resistance Distance},
  journal = {Journal of Mathematical Chemistry},
  volume  = {12},
  pages   = {81--95},
  year    = {1993},
  doi     = {10.1007/BF01164627}
}

@article{Sugiyama2023graph,
  author  = {Sugiyama, Tomohiro and Sato, Kazuhiro},
  title   = {{K}ron Reduction and Effective Resistance of Directed Graphs},
  journal = {SIAM Journal on Matrix Analysis and Applications},
  volume  = {44},
  number  = {1},
  pages   = {270--292},
  year    = {2023},
  doi     = {10.1137/22M1480823}
}

@article{Young2016,
  author  = {Young, George Forrest and
             Scardovi, Luca and
             Leonard, Naomi Ehrich},
  title   = {A New Notion of Effective Resistance for Directed Graphs---Part {II}: Computing Resistances},
  journal = {IEEE Transactions on Automatic Control},
  volume  = {61},
  number  = {7},
  pages   = {1737--1752},
  year    = {2016},
  doi     = {10.1109/TAC.2015.2481839}
}

@book{zhang2005schurComplement,
  author    = {Zhang, Fuzhen},
  title     = {The {S}chur Complement and Its Applications},
  publisher = {Springer},
  series    = {Numerical Methods and Algorithms},
  volume    = {4},
  address   = {New York, NY},
  year      = {2005},
  doi       = {10.1007/b105056},
  isbn      = {978-0-387-24271-2}
}
